\documentclass[a4paper,11pt]{article}
\usepackage{graphicx,epsfig}
\usepackage{fancyhdr,fancybox}
\usepackage{indentfirst}
\usepackage{verbatim}
\usepackage[sort&compress, numbers]{natbib}
\usepackage{geometry}
\usepackage{extarrows} % 画公式中的箭头
\usepackage[small]{caption2}
\usepackage{microtype}
\DisableLigatures[f]{encoding = *, family = *}

\usepackage{times}     % 包括 Times Roman + Helvetica + Courier
\usepackage{amssymb}                 % 用来排版漂亮的数学公式
\usepackage{amsthm}
\usepackage{mathrsfs}                % 英文花体字体
\usepackage{bbm}
\usepackage{hyperref}  % 生成书签

\newcommand{\paperfont}{\fontsize{12pt}{1.3\baselineskip}\selectfont}
\begin{document}
%\linenumbers

%%%%%%%%%% 定理类环境的定义 %%%%%%%%%%
\theoremstyle{definition}
\makeatletter
\thm@headfont{\bf}
\makeatother
\newtheorem{theorem}{Theorem}[section]
\newtheorem{definition}[theorem]{Definition}
\newtheorem{lemma}[theorem]{Lemma}
\newtheorem{proposition}[theorem]{Proposition}
\newtheorem{corollary}[theorem]{Corollary}
\newtheorem{remark}[theorem]{Remark}
\newtheorem{example}[theorem]{Example}
\newtheorem{assumption}[theorem]{Assumption}

%%%%%%%%%% 页眉和页脚的设置 %%%%%%%%%%
\lhead{}
\rhead{}
\lfoot{}
\rfoot{}

%%%%%%%%%% 一些重定义 %%%%%%%%%%
\renewcommand{\refname}{References}
\renewcommand{\figurename}{Figure}
\renewcommand{\tablename}{Table}
\renewcommand{\proofname}{Proof}

%%%%%%%%%% 符号重定义 %%%%%%%%%%
\newcommand{\diag}{\mathrm{diag}}
\newcommand{\tr}{\mathrm{tr}}
\newcommand{\re}{\mathrm{Re}}
\newcommand{\one}{\mathbbm{1}}
\newcommand{\loc}{\textrm{loc}}

% MATH -------------------------------------------------------------------
\newcommand{\Pnum}{\mathbb{P}}
\newcommand{\Enum}{\mathbb{E}}
\newcommand{\Rnum}{\mathbb{R}}
\newcommand{\Cnum}{\mathbb{C}}
\newcommand{\Znum}{\mathbb{Z}}
\newcommand{\Nnum}{\mathbb{N}}
\newcommand{\abs}[1]{\left\vert#1\right\vert}
\newcommand{\set}[1]{\left\{#1\right\}}
\newcommand{\norm}[1]{\left\Vert#1\right\Vert}
\newcommand{\innp}[1]{\langle {#1}\rangle}
\newcommand{\style}{\setlength{\itemsep}{1pt}\setlength{\parsep}{1pt}\setlength{\parskip}{1pt}}

%%%%%%%%%% 论文标题、作者等 %%%%%%%%%%
\title{\textbf{Mathematical foundation of nonequilibrium fluctuation-dissipation theorems for inhomogeneous diffusion processes with unbounded coefficients}}
\author{Xian Chen$^1$,\;\;\;Chen Jia$^{2,*}$ \\
\footnotesize $^1$ School of Mathematical Sciences, Xiamen University, Xiamen 361005, China.\\
\footnotesize $^2$ Department of Mathematics, Wayne State University, Detroit, Michigan 48202, U.S.A.\\
\footnotesize Correspondence: chenjia@wayne.edu}
\date{}                              % 日期
\maketitle                           % 生成标题
%\tableofcontents                    % 插入目录
\thispagestyle{empty}                % 首页无页眉页脚

%%%%%%%%%% 正式使用字体 %%%%%%%%%%%
\paperfont

%%%%%%%%%% 摘要 %%%%%%%%%%
\begin{abstract}
Nonequilibrium fluctuation-dissipation theorems (FDTs) are one of the most important advances in stochastic thermodynamics over the past two decades. Here we provide rigorous mathematical proofs of two types of nonequilibrium FDTs for inhomogeneous diffusion processes with unbounded drift and diffusion coefficients by using the Schauder estimates for partial differential equations of parabolic type and the theory of weakly continuous semigroups. The FDTs proved in this paper apply to any forms of inhomogeneous and nonlinear external perturbations. Furthermore, we prove the uniqueness of the conjugate observables and clarify the precise mathematical conditions and ranges of applicability for the two types of FDTs. Examples are also given to illustrate the main results of this paper.
\\

%%%%%%%%%% 关键词 %%%%%%%%%%
\noindent % 不缩进

%\noindent
%\textbf{Mathematics Subject Classifications}: 60J60, 82C05, 60H10, 82C31, 35K10
\end{abstract}

%%%%%%%%%% 正文 %%%%%%%%%%
\section{Introduction}
Over the past two decades, significant progress has been made in the field of mesoscopic nonequilibrium stochastic thermodynamics \cite{jarzynski2011equalities, seifert2012stochastic, van2015ensemble}, which has grown to become one the most important branches of statistical physics. The mathematical model of this theory turns out to be \emph{nonstationary and nonsymmetric Markov processes}, where the breaking of stationarity characterizes irreversibility described by Boltzmann and the breaking of symmetry characterizes irreversibility described by Prigogine \cite{esposito2010three, hong2016novel}. In this theory, an equilibrium state is defined as a stationary and symmetric Markov process and the deviation from equilibrium is usually characterized by the concept of entropy production \cite{zhang2012stochastic}. When an open system has a sustained external driving force, it will approach a nonequilibrium steady state (NESS), which is defined as a stationary but nonsymmetric Markov process.

The fluctuation-dissipation theorem (FDT) for equilibrium states is one of the classical results in statistical physics \cite{marconi2008fluctuation, kubo2012statistical, pavliotis2014stochastic}. In equilibrium, the FDT expresses the response of an observable to a small external perturbation by the correlation function of this observable and another one that is conjugate to the perturbation with respect to energy \footnote{In thermodynamics, intensive quantities such as temperature $T$, pressure $P$, and chemical potential $\mu$ are conjugate to extensive quantities such as entropy $S$, volume $V$, and particle number $N$, respectively, with respect to internal energy $U$.}. Mathematically, a small impulsive perturbation at time $s$ will give rise to a response of an observable $f$ at time $t$ that only depends on the time difference $t-s$ with the form of
\begin{equation*}
R_f(t-s) = -\frac{1}{k_BT}\frac{\partial}{\partial s}\Enum f(X_t)g(X_s),
\end{equation*}
where $k_B$ is the Boltzmann constant, $T$ is the temperature, and $g = -\partial_h|_{h=0}U$ is another observable conjugate to the perturbation $h$ with respect to energy $U$.

In recent years, numerous efforts have been devoted to generalizing the classical equilibrium FDT to systems far from equilibrium \cite{agarwal1972fluctuation,  lippiello2005off, speck2006restoring, lippiello2008nonlinear, chetrite2008fluctuation, chetrite2009fluctuation, seifert2010fluctuation, corberi2010fluctuation, hairer2010simple, dembo2010markovian, verley2011modified, verley2012fluctuations, chen2017fluctuation}. In fact, the study of the nonequilibrium FDT dated back to the work of Agarwal \cite{agarwal1972fluctuation}. In a recent paper, Seifert and Speck \cite{seifert2010fluctuation} have developed a new type of FDT based on concepts in stochastic thermodynamics. They found that in an NESS, the response of an observable to a small external perturbation can be represented by the correlation function of this observable and another one that is conjugate to the perturbation with respect to stochastic entropy. When a system is in equilibrium, stochastic entropy reduces to energy and the nonequilibrium FDT reduces to the classical equilibrium one. Moreover, the nonequilibrium FDT has been successfully applied to solve practical biological problems such as sensory adaptation in bacteria \cite{jia2017nonequilibrium} and gene regulation within cells \cite{yan2013fluctuation}.

We emphasize here that apart from the studies mentioned above, the nonequilibrium linear response has been investigated in at least two other important ways. One is the dynamical system approach initiated by Ruelle \cite{ruelle1998general, ruelle2004application, dolgopyat2004differentiability, ruelle2005differentiating, ruelle2008differentiation, baladi2017linear}, which focuses on the linear response for Sinai-Ruelle-Bowen measures --- one of the most important classes of invariant measures for dissipative dynamical systems with chaotic behavior. The other is the path-integral approach initiated by Baiesi, Maes, and Wynants \cite{baiesi2009fluctuations, baiesi2009nonequilibrium, baiesi2010nonequilibrium, maes2010response, baiesi2013update}, which identifies the entropic and frenetic contributions to the linear response. Although these two approaches play a crucial role in understanding the fluctuations and response of nonequilibrium states, they are beyond the scope of the present paper.

The mathematical theory for NESSs has been developed for more than three decades \cite{qian1982circulation, jiang2004mathematical, ge2008generalized, jia2015second, wang2016asymptotics, jia2016model, jia2016cycle, ge2017cycle}. However, this theory cannot be directly applied here because the FDTs focus on the nonstationary perturbation and the time-dependent dynamic behavior of a Markov process. In the physics literature, the derivation of the FDTs is formal and not rigorous. In a recent work of Dembo and Deuschel \cite{dembo2010markovian}, a type of FDT has been rigorously proved for general homogenous Markov processes in an abstract setting. For inhomogeneous Markov jump processes with discrete state spaces, an attempt has also been made to integrate the FDTs into a rigorous mathematical framework \cite{chen2017fluctuation}. However, for inhomogeneous Markov processes with continuous state spaces, the proof of the FDTs turns out to be highly nontrivial due to the lack of effective mathematical tools.

In stochastic thermodynamics, the most important mathematical model of a molecular system is the diffusion process \cite{revuz1999continuous}, which generalizes the classical Langevin equation describing the stochastic movement of multiple massive particles in a fluid due to collisions with the fluid molecules. As a beneficial attempt, the Jarzynski-Crooks work relation has been rigorously proved for diffusion processes \cite{ge2008generalized}. However, this proof requires that both the drift and diffusion matrix of a diffusion process are bounded and all their partial derivatives are bounded. These assumptions are so strong that they even exclude the classical Ornstein-Uhlenbeck (OU) process, which describes the velocity of a massive Brownian particle under the influence of friction \cite{uhlenbeck1930theory}.

In fact, the Kolmogorov forward and backward equations for a diffusion process are partial differential equations of parabolic type. The reason why such strong bounded assumptions are made is that they serve as the basic requirements of the classical parabolic equation theory \cite{friedman1964partial, lieberman1996second}. In the present paper, we remove these bounded assumptions and provide a rigorous mathematical foundation of two types of FDTs | the Agarwal-type and the Seifert-Speck-type FDTs | for \emph{inhomogeneous diffusion processes with unbounded coefficients}. It turns out that the theory of this paper applies to any form of nonlinear external perturbations, rather than merely linear perturbations as in most previous papers. Furthermore, we also prove the uniqueness of the conjugate observables and clarify the mathematical conditions and ranges of applicability for the two types of FDTs.

The structure of the present work is organized as follows. In Section 2, we introduce the fundamental framework of the FDTs for inhomogeneous diffusion processes and make the basic assumptions. In particular, we give the rigorous mathematical definition of the perturbed process and response function. In Section 3, we provide a mathematical theory of the Agarwal-type FDT for inhomogeneous diffusion processes using the Schauder estimates for parabolic equations. In Section 4, we provide a mathematical theory of the Seifert-Speck-type FDT for homogenous diffusion processes using the theory of weakly continuous semigroups. Section 5 is devoted to clarifying the relationship between the two types of FDTs. In Section 6, we use the example of inhomogeneous OU processes to illustrate the main results.

\section{Model}

\subsection{Model and basic assumptions}
Let $W = \{W_t:t\geq 0\}$ be an $n$-dimensional standard Brownian motion defined on a filtered probability space $(\Omega,\mathcal{F},P)$. In this paper, we consider a molecular system modeled by a $d$-dimensional \emph{inhomogeneous diffusion processes} $X = \{X_t:t\geq 0\}$, which is the solution to the stochastic differential equation (SDE)
\begin{equation}\label{SDE}
dX_t = b(t,X_t)dt+\sigma(t,X_t)dW_t,
\end{equation}
where $b:\Rnum^+\times\Rnum^d\rightarrow\Rnum^d$ and $\sigma:\Rnum^+\times\Rnum^d\rightarrow M_{d\times n}(\Rnum)$ with $M_{d\times n}(\Rnum)$ being the vector space of all $d\times n$ real matrices. Then $X$ is called a diffusion process with drift $b = (b^i)$ and diffusion matrix $a = \sigma\sigma^T = (a^{ij})$. We shall give conditions ensuring the existence and uniqueness of the process in the following discussion. Recall that $X$ is called \emph{homogenous} if the drift $b =  b(t,x)$ and diffusion matrix $a =  a(t,x)$ only depend on the spatial variable $x$ and do not depend on the time variable $t$. Otherwise, $X$ is called \emph{inhomogeneous}. If $\sigma$ is a constant diagonal matrix, then the SDE \eqref{SDE} is also called a \emph{Langevin equation}. Most previous studies \cite{marconi2008fluctuation, chetrite2008fluctuation, dembo2010markovian} focused on the response of a homogeneous Langevin equation to a small external perturbation. Here we consider the response of a general inhomogeneous diffusion process.

Following standard notation, for any multi-index $\beta = (\beta_1,\cdots,\beta_d)$, set $|\beta| = \beta_1+\cdots+\beta_d$ and $D^\beta = \partial_1^{\beta_1}\cdots\partial_d^{\beta_d}$, where $\partial_i = \partial/\partial_i$ denotes the $i$th weak or strong partial derivative. In this paper, we need the following function spaces.

\begin{definition}\label{functionspace}
In the following function spaces, $[0,T]$ and $\Rnum^d$ may also be replaced by subsets of $[0,T]$ and $\Rnum^d$, respectively.
\begin{itemize}\style
\item Let $B(\Rnum^d)$ denote the space of all bounded measurable functions on $\Rnum^d$.
\item Let $C_c^k(\Rnum^d)$ with $k\in\Znum^+\cup\{\infty\}$ denote the space of all $f\in C^k(\Rnum^d)$ with compact support.
\item Let $C_b(\Rnum^d)$ denote the Banach space of all bounded continuous functions on $\Rnum^d$ endowed with the supremum norm $\|\cdot\|$.
\item Let $C^k_b(\Rnum^d)$ with $k\in\Znum^+$ denote the Banach space of all $f\in C^k(\Rnum^d)$ such that
    \begin{equation*}
    \|f\|_{C^k_b(\Rnum^d)} := \sum_{|\beta|\leq k}\|D^\beta f\| < \infty.
    \end{equation*}
\item Let $C^\alpha_b(\Rnum^d)$ with $0<\alpha<1$ denote the Banach space of all bounded $\alpha$-H\"{o}lder continuous functions on $\Rnum^d$ with norm defined as
    \begin{equation*}
    \|f\|_{C^\alpha_b(\Rnum^d)} = \|f\|+\sup_{x,y\in\Rnum^d,x\neq y}\frac{|f(x)-f(y)|}{|x-y|^\alpha}.
    \end{equation*}
\item Let $C^\alpha_{\loc}(\Rnum^d)$ with $0<\alpha<1$ denote the space of all measurable functions $f$ on $\Rnum^d$ such that $f\in C^\alpha_b(U)$ for any bounded open subsets $U\subset\Rnum^d$.
\item Let $C^{k+\alpha}_b(\Rnum^d)$ with $k\in\Znum^+$ and $0<\alpha<1$ denote the Banach space of all $f\in C^k_b(\Rnum^d)$ whose all $k$th-order partial derivatives are $\alpha$-H\"{o}lder continuous with norm defined as
    \begin{equation*}
    \|f\|_{C^{k+\alpha}_b(\Rnum^d)} = \|f\|_{C^k_b(\Rnum^d)}+\sum_{|\beta|=k}\|D^\beta f\|_{C^\alpha_b(\Rnum^d)}.
    \end{equation*}
\item Let $C^{l,k}([0,T]\times\Rnum^d)$ with $l,k\in\mathbb{Z}^+$ denote the space of all continuous functions $f$ on $[0,T]\times\Rnum^d$ that are $l$th-order continuously differentiable with respect to the time variable and $k$th-order continuously differentiable with respect to the spatial variable.
\item Let $C^{0,k+\alpha}_b([0,T]\times\Rnum^d)$ with $k\in\Znum^+$ and $0<\alpha<1$ denote the Banach space of all continuous functions $f$ on $[0,T]\times\Rnum^d$ such that
    \begin{equation*}
    \|f\|_{C^{0,k+\alpha}_b([0,T]\times\Rnum^d)} := \sup_{t\in[0,T]}\|f(t,\cdot)\|_{C^{k+\alpha}_b(\Rnum^d)} < \infty.
    \end{equation*}
\item Let $C^{0,0,k+\alpha}_b([0,T]\times[-1,1]\times\Rnum^d)$ with $k\in\Znum^+$ and $0<\alpha<1$ denote the Banach space of all continuous functions $f$ on $[0,T]\times[-1,1]\times\Rnum^d$ such that
    \begin{equation*}
    \|f\|_{C^{0,0,k+\alpha}_b([0,T]\times[-1,1]\times\Rnum^d)} := \sup_{(t,h)\in[0,T]\times[-1,1]}\|f(t,h,\cdot)\|_{C^{k+\alpha}_b(\Rnum^d)} < \infty.
    \end{equation*}
\item Let $L^p_{\loc}(\Rnum^d)$ with $p\geq 1$ denote the space of all measurable functions $f$ on $\Rnum^d$ such that $f\in L^p(U)$ for any bounded open subset $U\subset\Rnum^d$.
\item Let $W^{k,p}(\Rnum^d)$ with $k\in\mathbb{Z}^+$ and $p\geq 1$ denote the Banach space of all $k$th-order weakly differentiable functions $f$ on $\Rnum^d$ such that
    \begin{equation*}
    \|f\|_{W^{k,p}(\Rnum^d)} := \sum_{|\beta|\leq k}\|D^\beta f\|_{L^p(\Rnum^d)} < \infty.
    \end{equation*}
    The space $W^{k,2}(\Rnum^d)$ is also denoted by $H^k(\Rnum^d)$.
\item Let $W^{k,p}_{\loc}(\Rnum^d)$ with $k\in\mathbb{Z}^+$ and $p\geq 1$ denote the space of all measurable functions $f$ on $\Rnum^d$ such that $f\in W^{k,p}(U)$ for any bounded open subsets $U\subset\Rnum^d$. The space $W^{k,2}_{\loc}(\Rnum^d)$ is also denoted by $H^k_{\loc}(\Rnum^d)$.
\item For any open subset $U\subset\Rnum^d$, let $W^{k,p}_{0}(U)$ with $k\in\mathbb{Z}^+$ and $p\geq 1$ denote the closure of $C_c^{\infty}(U)$ in $W^{k,p}(U)$. The space $W^{k,2}_0(U)$ is also denoted by $H^k_0(U)$. We stress here that if we take $U = \Rnum^d$, then $W^{k,p}_0(\Rnum^d) = W^{k,p}(\Rnum^d)$ \cite[Corollary 3.23]{adams2003sobolev}. However, This equality may not hold for general open subset $U$.
\end{itemize}
\end{definition}

Recall that the inhomogeneous diffusion process $X$ is associated with a family of second-order elliptic operators $\{\mathcal{A}_t:t\geq 0\}$ defined by
\begin{equation}\label{generator}
\mathcal{A}_tf = \sum_{i=1}^{d}b^i(t,x)\partial_if+\frac{1}{2}\sum_{i,j=1}^{d}a^{ij}(t,x)\partial_{ij}f,\;\;\;f\in W^{2,1}_{\loc}(\Rnum^d).
\end{equation}
The transition semigroup $\{P_{s,t}:0\leq s\leq t\}$ of $X$ is defined as
\begin{equation}\label{semigroup}
P_{s,t}f(x) = \mathbb{E}_{s,x}f(X_t) := \Enum\{f(X_t)|X_s=x\},\;\;\;f\in B(\Rnum^d).
\end{equation}
In addition, recall that the operator norm of a $d\times d$ real matrix $A = (A^{ij})$ is defined as
\begin{equation*}
|A| = \sup_{x\neq 0}\frac{|Ax|}{|x|}.
\end{equation*}
It is easy to check that the operator norm is controlled by the Frobenius norm:
\begin{equation*}
|A| \leq \left(\sum_{i,j=1}^dA_{ij}^2\right)^{1/2}.
\end{equation*}

In the following discussion, we make the convention that when we say that a vector-valued or matrix-valued function belongs to a particular function space, we mean that all the entries of this function do belong. Moreover, we always fix a time $T>0$ and \emph{consider the dynamics of $X$ up to time $T$}.

\begin{definition}\label{regular}
We say that $X$ satisfies the \emph{regular conditions} if there exist $0<\alpha<1$, two constants $\lambda,C>0$, and a function $\eta:[0,T]\times\Rnum^d\rightarrow\Rnum$ such that the following five conditions hold:
\begin{itemize}\style
\item[(a)] $b,a\in C^{0,3+\alpha}_b([0,T]\times B_R)$ for each $R>0$, where $B_R = \{x\in\Rnum^d:|x|<R\}$ is the ball in $\Rnum^d$ with radius $R$ centered at the origin.
\item[(b)] The diffusion matrix $a$ satisfies
    \begin{equation*}
    \xi^Ta(t,x)\xi \geq \eta(t,x)|\xi|^2,\;\;\;\forall\;t\in[0,T],\xi,x\in\Rnum^d,
    \end{equation*}
    where
    \begin{equation*}
    \inf_{(t,x)\in[0,T]\times\Rnum^d}\eta(t,x)\geq \lambda > 0.
    \end{equation*}
\item[(c)] The drift $b$ and diffusion matrix $a$ are controlled by
    \begin{equation*}
    \begin{split}
    b(t,x)^Tx &\leq C\eta(t,x)(1+|x|^2),\\
    |a(t,x)x|+\tr(a(t,x)) &\leq C\eta(t,x)(1+|x|^2),\;\;\;\forall\;t\in[0,T],x\in\Rnum^d.
    \end{split}
    \end{equation*}
\item[(d)] For any multi-index $\beta = (\beta_1,\cdots,\beta_d)$ with $1\leq |\beta|\leq 3$,
    \begin{equation*}
    |D^{\beta}b(t,x)|+|D^{\beta}a(t,x)|\leq C\eta(t,x),\;\;\;\forall\;t\in[0,T],x\in\Rnum^d,
    \end{equation*}
\item[(e)] There exists a function $\psi\in C^2(\Rnum^d)$ satisfying $\psi(x)\rightarrow\infty$ as $|x|\rightarrow\infty$ such that
    \begin{equation*}
    \mathcal{A}_t\psi(x)\leq C\psi(x),\;\;\;\forall\;t\in[0,T],x\in\Rnum^d.
    \end{equation*}
\end{itemize}
\end{definition}

Here the function $\psi$ is called the \emph{Lyapunov function}. If $X$ is homogenous, then the regular condition (c) can be removed \cite[Theorem 2]{lunardi1998schauder}.

\begin{remark}
It turns out that the Schauder estimates for parabolic equations will play a crucial role in the proofs of the FDTs. In fact, the classical theory of Schauder estimates for elliptic and parabolic equations focuses on the case when $b$ and $a$ are bounded \cite{friedman1964partial, lieberman1996second}. If $b$ and $a$, together with all their spatial partial derivatives, are bounded and continuous, then the regular conditions (a), (c), and (d) are automatically satisfied. If we take $\psi(x) = 1+|x|^2$, then the regular condition (e) is also satisfied. The major reason for us to impose the regular conditions (a)-(e) is to ensure the Schauder estimates for parabolic equations with unbounded coefficients \cite{lunardi1998schauder, lorenzi2011optimal}.
\end{remark}

\begin{remark}
If $X$ is homogenous, then $b$, $a$, and $\eta$ only depend on the spatial variable $x$ and thus the constants $\alpha$, $\lambda$, and $C$ do not depend on the time $T$. In this case, we do not need to fix the time $T$ and all the results of the present paper will not change if we replace $T$ by $\infty$.
\end{remark}

In the following discussion, unless otherwise specified, we always assume that the regular conditions (a)-(e) are satisfied. For each $R>0$, let $\tau_R = \inf\{t\geq 0:X_t\in\partial B_R\}$ be the hitting time of the sphere $\partial B_R$ by $X$. Recall that the explosion time $\tau$ of $X$ is defined as
\begin{equation*}
\tau = \lim_{R\rightarrow\infty}\tau_R.
\end{equation*}
The following proposition shows that regular conditions guarantee the existence, uniqueness, and nonexplosiveness of $X$.

\begin{proposition}
Suppose that the regular conditions (a), (b), and (e) are satisfied. Then the following two statements hold:
\begin{itemize}\style
\item[(a)] The weak solution of \eqref{SDE} exists up to time $T$ and is unique in law. In particular, $X$ is nonexplosive up to time $T$.
\item[(b)] If $n = d$ and $\sigma = a^{1/2}$, then the strong solution of \eqref{SDE} exists up to time $T$ and is pathwise unique.
\end{itemize}
\end{proposition}

\begin{proof}
We first prove that if the solution of \eqref{SDE} exists, it must be nonexplosive up to time $T$. To this end, let $\psi$ be the Lyapunov function in the regular condition (e). By Ito's formula, we have
\begin{equation*}
d\psi(X_s) = \mathcal{A}_s\psi(X_s)ds+\nabla\psi(X_s)^T\sigma(s,X_s)dW_s.
\end{equation*}
For any $0\leq t\leq T$ and $|x|<R$, it follows from the regular condition (e) that
\begin{equation*}
\Enum_x\psi(X_{t\wedge\tau_R}) = \phi(x)+\Enum_x\int_0^t\mathcal{A}_s\psi(X_s)I_{\{s\leq\tau_R\}}ds
\leq \psi(x)+C\int_0^t\Enum_x\psi(X_{s\wedge\tau_R})ds.
\end{equation*}
By Gronwall's inequality, we have
\begin{equation*}
\Enum_x\psi(X_{\tau_R})I_{\{\tau_R\leq t\}} \leq \Enum_x\psi(X_{t\wedge \tau_R})
\leq \psi(x)e^{Ct},\;\;\;\forall\;t\in[0,T].
\end{equation*}
This shows that
\begin{equation*}
\min_{|y|=R}\psi(y)\cdot P_x(\tau_R\leq T)\leq \psi(x)e^{CT}.
\end{equation*}
Since $\psi(y)\rightarrow\infty$ as $|y|\rightarrow\infty$, we have $P_x(\tau_R\leq T)\rightarrow 0$ as $R\rightarrow\infty$. This indicates that
\begin{equation*}
P_x\left(\tau\leq T\right) = \lim_{R\rightarrow\infty}P_x(\tau_R\leq T) = 0.
\end{equation*}
This shows that $X$ is nonexplosive, that is, $\tau>T$ almost surely.

We next prove (b). The regular condition (a) implies that $b$ and $a$ are locally Lipschitz up to time $T$: for any $R>0$, there exists a constant $K>0$ such that
\begin{equation*}
|b(t,x)-b(t,y)|\leq K|x-y|,\;\;\;|a(t,x)-a(t,y)|\leq K|x-y|,\;\;\;\forall\;t\in[0,T],x,y\in B_R.
\end{equation*}
By \cite[Theorem 5.2.2]{stroock2006multidimensional}, the regular condition (b) implies that
\begin{equation*}
|a^{1/2}(t,x)-a^{1/2}(t,y)|\leq \frac{K}{2\lambda^{1/2}}|x-y|,\;\;\;\forall\;t\in[0,T],x\in B_R.
\end{equation*}
This shows that $\sigma$ is also locally Lipschitz up to time $T$. Since $b$ and $\sigma$ are locally Lipschitz and any solution of \eqref{SDE} is nonexplosive up to time $T$, Ito's existence and uniqueness theorem gives the desired result.

We finally prove (a). Since the strong solution of \eqref{SDE} exists up to time $T$ and is pathwise unique when $\sigma = a^{1/2}$, it is easy to see that the martingale problem for ($a,b$) is well-posed up to time $T$. By the equivalence between the martingale problem formulation and the weak solution formulation \cite[Theorem 20.1]{rogers2000diffusions2}, the weak solution of \eqref{SDE} exists up to time $T$ and is unique in law.
\end{proof}

\subsection{Perturbed processes}
We next investigate the response of $X$ to a small external perturbation. For any $h\in C[0,T]$ with $\|h\|\leq 1$ which characterizes the perturbation protocol, we consider another inhomogeneous diffusion process $X^h = \{X^h_t:t\geq 0\}$ with perturbed drift $b_h:[0,T]\times\Rnum^d\rightarrow\Rnum^d$ and diffusion matrix $a_h:[0,T]\times\Rnum^d\rightarrow M_{d\times d}(\Rnum)$. Since $X^h$ is a perturbation of $X$, it is natural to assume that they \emph{have the same initial distribution} and there exist two trivariate functions
\begin{equation}\label{trivariate}
\bar b:[0,T]\times[-1,1]\times\Rnum^d\rightarrow \Rnum^d,\;\;\;
\bar a:[0,T]\times[-1,1]\times\Rnum^d\rightarrow M_{d\times d}(\Rnum)
\end{equation}
such that
\begin{gather*}
b(t,x) = \bar b(t,0,x),\;\;\;a(t,x) = \bar a(t,0,x), \\
b_h(t,x) = \bar b(t,h(t),x),\;\;\;a_h(t,x) = \bar a(t,h(t),x),
\end{gather*}
where $t$ is the time variable, $h$ is the perturbation variable, and $x$ is the spatial variable. In analogy to \eqref{generator} and \eqref{semigroup}, we can define the second-order elliptic operators $\{\mathcal{A}^h_t\}$ and transition semigroup $\{P^h_{s,t}\}$ associated with $X^h$. In the following discussion, we do not distinguish $b_h(t,x)$ from $\bar b(t,h,x)$ and do not distinguish $a_h(t,x)$ from $\bar a(t,h,x)$. The meaning should be clear from the context.

\begin{remark}
In \cite[Section 2.3.2]{marconi2008fluctuation}, the authors added a linear perturbation to the drift and kept the diffusion matrix unchanged. In addition, the authors assumed that the perturbed drift has the form of
\begin{equation*}
b_h(t,x) = b(t,x)+hF(t)K(x).
\end{equation*}
Here we remove these two restrictions and consider a general nonlinear external perturbation.
\end{remark}

To proceed, we write the perturbed drift $b_h$ and diffusion matrix $a_h$ as
\begin{equation*}
b_h(t,x) = b(t,x)+hq_h(t,x),\;\;\;a_h(t,x) = a(t,x)+hr_h(t,x),
\end{equation*}
where $q_h = (q_h^i)$ and $r_h = (r_h^{ij})$. We assume that $b_h$ and $a_h$ are differentiable with respect to $h$. For convenience, we write
\begin{equation*}
q(t,x) = \partial_h|_{h=0}b_h(t,x),\;\;\;r(t,x) = \partial_h|_{h=0}a_h(t,x),
\end{equation*}
where $q = (q^i)$ and $r = (r^{ij})$. For any $0\leq t\leq T$, we define the following second-order differential operators:
\begin{equation*}
\begin{split}
\mathcal{L}^h_tf &= \sum_{i=1}^{d}q_h^i(t,x)\partial_if+\frac{1}{2}\sum_{i,j=1}^{d}r_h^{ij}(t,x)\partial_{ij}f,\;\;\;f\in W^{2,1}_{\loc}(\Rnum^d),\\
\mathcal{L}_tf &= \sum_{i=1}^{d}q^i(t,x)\partial_if+\frac{1}{2}\sum_{i,j=1}^{d}r^{ij}(t,x)\partial_{ij}f,\;\;\;f\in W^{2,1}_{\loc}(\Rnum^d).
\end{split}
\end{equation*}

\begin{assumption}\label{ass}
Throughout this paper, we assume that there exist two constants $0<\theta<1$ and $L>0$ such as the following four conditions hold:
\begin{itemize}\style
\item[(a)] $\bar b,\bar a\in C^{0,0,3+\alpha}_b([0,T]\times[-1,1]\times B_R)$ for each $R>0$.
\item[(b)] The functions $q_h$ and $r_h$ are controlled by
    \begin{equation*}
    \|q_h\|_{C^{0,\theta}_b([0,T]\times\Rnum^d)}
    +\|r_h\|_{C^{0,\theta}_b([0,T]\times\Rnum^d)}\leq L,\;\;\;\forall\;h\in[-1,1],
    \end{equation*}
    where
    \begin{equation*}
    \begin{split}
    \|q_h\|_{C^{0,\theta}_b([0,T]\times\Rnum^d)} &= \left(\sum_{i=1}^d\|q_h^i\|_{C^{0,\theta}_b([0,T]\times\Rnum^d)}^2\right)^{1/2},\\
    \|r_h\|_{C^{0,\theta}_b([0,T]\times\Rnum^d)} &=
    \left(\sum_{i,j=1}^d\|r_h^{ij}\|_{C^{0,\theta}_b([0,T]\times\Rnum^d)}^2\right)^{1/2}.
    \end{split}
    \end{equation*}
\item[(c)] For any multi-index $\beta = (\beta_1,\cdots,\beta_d)$ with $1\leq|\beta|\leq 3$,
    \begin{equation*}
    |D^\beta q_h(t,x)|+|D^\beta r_h(t,x)|\leq L\eta(t,x),\;\;\;\forall\;t\in[0,T],h\in[-1,1],x\in\Rnum^d.
    \end{equation*}
    where $\eta(t,x)$ is the function introduced in Definition \ref{regular}.
\item[(d)] The Lyapunov function $\psi$ defined in the regular condition (e) satisfies
    \begin{equation*}
    \mathcal{L}^h_t\psi(x)\leq L\psi(x),\;\;\;\forall\;t\in[0,T],h\in[-1,1],x\in\Rnum^d.
    \end{equation*}
\end{itemize}
\end{assumption}

The following lemma follows directly from the above assumptions.

\begin{lemma}\label{perturbation}
When $\|h\|$ is sufficiently small, the perturbed process $X^h$ also satisfies the regular conditions and the constants $\lambda$ and $C$ in the regular conditions can be chosen to be independent of $h$.
\end{lemma}

\begin{proof}
By Assumption \ref{ass}(a), it is easy to see that for any $R>0$,
\begin{equation*}
\sup_{t\in[0,T]}\|b_h(t,\cdot)\|_{C^{3+\alpha}_b(B_R)}
\leq \sup_{(t,h)\in[0,T]\times[-1,1]}\|\bar b(t,h,\cdot)\|_{C^{3+\alpha}_b(B_R)} < \infty,
\end{equation*}
which shows that $b_h\in C^{0,3+\alpha}_b([0,T]\times B_R)$. Similarly, we also have $a_h\in C^{0,3+\alpha}_b([0,T]\times B_R)$. Thus $X^h$ satisfies the regular condition (a). By Assumption \ref{ass}(b), when $\|h\|$ is sufficiently small, for any $0\leq t\leq T$ and $\xi,x\in\Rnum^d$,
\begin{equation*}
\begin{split}
\xi^Ta_h(t,x)\xi &\geq \xi^Ta(t,x)\xi-|a_h(t,x)-a(t,x)\|\xi|^2
\geq (\eta(t,x)-L\|h\|)|\xi|^2 \geq \frac{1}{2}\eta(t,x)|\xi|^2.
\end{split}
\end{equation*}
This shows that $X^h$ satisfies the regular condition (b) and the constant $\lambda$ can be chosen to be independent of $h$. By using similar techniques, it is easy to prove that $X^h$ satisfies the regular condition (c)-(e) and the constant $C$ can be chosen to be independent of $h$.
\end{proof}

\subsection{Response function}
In order to give the rigorous definition for the response function, we recall the definition of the functional derivative.
\begin{definition}
Fix $t>0$. Let $F$ be a functional on $C[0,t]$ and let $h\in C[0,t]$. Then the functional derivative of $F$ with respect to $h$ is a functional $\delta F/\delta h$ on $C_c^\infty(0,t)$ defined as
\begin{equation*}
\langle\frac{\delta F}{\delta h},\phi\rangle := \frac{d}{d\epsilon}\Big|_{\epsilon = 0}F(h+\epsilon\phi) := \lim_{\epsilon\rightarrow 0}\frac{1}{\epsilon}(F(h+\epsilon\phi)-F(h)),
\end{equation*}
whenever the limit exists for any $\phi\in C_c^\infty(0,t)$.
\end{definition}

We next define the response function of an observable.
\begin{definition}\label{defresponse}
Let $f:\Rnum^d\rightarrow\Rnum$ be a bounded observable. For any $0\leq t\leq T$, let $F_t$ be a functional on $C[0,T]$ defined as
\begin{equation*}
F_t(h) = \Enum f(X^h_t).
\end{equation*}
If for any $0\leq t\leq T$, there exists a locally integrable function $R_f(\cdot,t)$ on $(0,t)$ such that
\begin{equation}\label{response}
\langle\frac{\delta F_t}{\delta h}\Big|_{h=0},\phi\rangle = \int_0^tR_f(s,t)\phi(s)ds,\;\;\;\forall \phi\in C_c^\infty(0,t),
\end{equation}
then $R_f(s,t)$ is called the \emph{response function} of the observable $f$.
\end{definition}

The physical implication of the response function $R_f(s,t)$ can be understood as follows. Formally, if we take $\phi(x) = \delta_s(x) = \delta(x-s)$ in \eqref{response}, then we have
\begin{equation*}
R_f(s,t) = \int_0^tR_f(u,t)\delta_s(u)du = \langle\frac{\delta F_t}{\delta h}\Big|_{h=0},\delta_s\rangle = \lim_{\epsilon\rightarrow 0}\frac{1}{\epsilon}(F_t(\epsilon\delta_s)-F_t(0)).
\end{equation*}
This suggests that if we add a small impulsive perturbation $\epsilon\delta_s$ to $X$ at time $s$, then the rate of change for the ensemble average at time $t$ is exactly $R_f(s,t)$.

\section{The Agarwal-type FDT}
We first study the Agarwal-type FDT. Some of the lemmas in the following two sections are well known in the case of homogeneous or bounded coefficients, while they are nontrivial in the case of inhomogeneous and unbounded coefficients. Unless otherwise specified, we always assume that the regular conditions (a)-(e) and Assumption \ref{ass} are satisfied. The following lemma characterizes the evolution of $X$.

\begin{lemma}\label{kolmogorov}
For any $f\in C_b(\Rnum^d)$ and $0\leq s\leq t\leq T$, the function $v(s,x) = P_{s,t}f(x)\in C_b([0,t]\times\Rnum^d)\cap C^{1,2}([0,t)\times\Rnum^d)$ is the unique bounded classical solution to the following parabolic equation, which is also called the Kolmogorov backward equation:
\begin{equation}\label{backward}
\begin{cases}
\partial_sv = -\mathcal{A}_sv,\;\;\;0\leq s< t,\\
v(t,x) = f(x).
\end{cases}
\end{equation}
\end{lemma}

\begin{proof}
Under the regular conditions (a)-(e), by \cite[Theorems 2.7 and 3.8 and Remark 2.8]{lorenzi2011optimal}, the above parabolic equation has a unique bounded classical solution $v\in C_b([0,t]\times\Rnum^d)\cap C^{1,2}([0,t)\times\Rnum^d)$. It follows from Ito's formula that
\begin{equation*}
\begin{split}
dv(u,X_u) &= [\partial_uv(u,X_u)+\mathcal{A}_uv(u,X_u)]du+\nabla v(u,X_u)^T\sigma(u,X_u)dW_u \\
&= \nabla v(u,X_u)^T\sigma(u,X_u)dW_u.
\end{split}
\end{equation*}
For any $|x|<R$ and $s<r<t$, we have
\begin{equation*}
v(s,x) = \Enum_{s,x}v(r\wedge \tau_R,X_{r\wedge \tau_R})-\Enum_{s,x}\int_s^r\nabla v(u,X_u)^T\sigma(u,X_u)I_{\{u\leq \tau_R\}}dW_u.
\end{equation*}
The fact that $v\in C^{1,2}([0,t)\times\Rnum^d)$ and the regular condition (a) indicate that $\nabla v$ and $a$ are continuous function on $[0,r]\times B_R$. It thus follows from \cite[Chapter IV, Corollary 1.25]{revuz1999continuous} that
\begin{equation*}
\Enum_{s,x}\int_s^r\nabla v(u,X_u)^T\sigma(u,X_u)I_{\{u\leq \tau_R\}}dW_u = 0.
\end{equation*}
Since $X$ is nonexplosive up to time $T$ and $v\in C_b([0,t]\times\Rnum^d)$, we have
\begin{equation*}
v(s,x) = \lim_{r\rightarrow t}\lim_{R\rightarrow\infty}\Enum_{s,x}v(r\wedge \tau_R,X_{r\wedge \tau_R})
= \lim_{r\rightarrow t}\Enum_{s,x}v(r,X_r) = \Enum_{s,x}v(t,X_t) = \Enum_{s,x}f(X_t),
\end{equation*}
which gives the desired result.
\end{proof}

\begin{remark}
If $f\in C^{2+\theta}_b(\Rnum^d)$ for some $0<\theta<1$, then it can be proved that the function $v(s,x) = P_{s,t}f(x)$ is the unique bounded classical solution to the Kolmogorov backward equation:
\begin{equation*}
\begin{cases}
\partial_sv = -\mathcal{A}_sv,\;\;\;0\leq s\leq t,\\
v(t,x) = f(x).
\end{cases}
\end{equation*}
where $s = t$ is included. In this case, we have $v\in C^{1,2}([0,t]\times\Rnum^d)$ \cite[Theorems 2.7 and Remark 2.8]{lorenzi2011optimal}.
\end{remark}

The following lemma gives the semigroup estimates for $X$.
\begin{lemma}\label{contractive}
For any $0\leq\gamma\leq 3$, there exists a constant $K = K(d,T,\gamma,\lambda,C)>0$ such that
\begin{equation*}
\|P_{s,t}f\|_{C^\gamma_b(\Rnum^d)}\leq K\|f\|_{C^\gamma_b(\Rnum^d)},\;\;\;
\forall\;f\in C^\gamma_b(\Rnum^d),0\leq s\leq t\leq T,
\end{equation*}
where $\lambda$ and $C$ are the two constants introduced in Definition \ref{regular}.
\end{lemma}

\begin{proof}
Since $v(s,x) = P_{s,t}f(x)$ satisfies the Kolmogorov backward equation \eqref{backward}, the desired result follows from \cite[Theorem 2.4]{lorenzi2011optimal}.
\end{proof}

The following theorem is interesting in its own right.

\begin{theorem}\label{gestimation}
Fix $0\leq t\leq T$, $h\in C[0,T]$, and $f\in C^2_b(\Rnum^d)$. Let $g$ be a function on $[0,t]\times\Rnum^d$ defined by
\begin{equation*}
g(s,x) := \mathcal{L}^h_sP_{s,t}^hf(x).
\end{equation*}
Then there exists a constant $K = K(d,T,\theta,\lambda,C)>0$ such that the following two statements hold when $\|h\|$ is sufficiently small:
\begin{itemize}\style
\item[(a)] For any $f\in C^2_b(\Rnum^d)$, we have $g\in C_b([0,t]\times\Rnum^d)$ and
    \begin{equation*}
    \|g\|_{C_b([0,t]\times\Rnum^d)}\leq 2KL\|f\|_{C^2_b(\Rnum^d)},
    \end{equation*}
    where $L$ is the constant introduced in Assumption \ref{ass}.
\item[(b)] For any $f\in C^{2+\theta}_b(\Rnum^d)$, we have $g\in C^{0,\theta}_b([0,t]\times\Rnum^d)$ and
    \begin{equation*}
    \|g\|_{C^{0,\theta}_b([0,t]\times\Rnum^d)}\leq 4KL\|f\|_{C^{2+\theta}_b(\Rnum^d)}.
    \end{equation*}
\end{itemize}
\end{theorem}

\begin{proof}
It is easy to see that
\begin{equation*}
g(s,x) = \sum_{i=1}^dq^i_h(s,x)\partial_iP_{s,t}^hf(x) +\frac{1}{2}\sum_{i,j=1}^dr^{ij}_h(s,x)\partial_{ij}P_{s,t}^hf(x) := g_1(s,x)+g_2(s,x).
\end{equation*}
It follows from Assumption \ref{ass}(b) that
\begin{equation*}
\|g\|_{C_b([0,t]\times\Rnum^d)}\leq 2L\sup_{0\leq s\leq t}\|P_{s,t}^hf\|_{C^2_b(\Rnum^d)}.
\end{equation*}
By Lemmas \ref{perturbation} and \ref{contractive}, there exists a constant $K = K(d,T,\lambda,C)>0$ such that when $\|h\|$ is sufficiently small,
\begin{equation*}
\|P^h_{s,t}f\|_{C_b^2(\Rnum^d)}\leq K\|f\|_{C_b^2(\Rnum^d)},\;\;\;\forall\;0\leq s\leq t\leq T.
\end{equation*}
Thus we have proved (a). In addition, it is easy to check that
\begin{equation*}
\begin{split}
&\;g_1(s,x)-g_1(s,y)\\
=&\; \sum_{i=1}^dq^i_h(s,x)(\partial_iP_{s,t}^hf(x)-\partial_iP_{s,t}^hf(y))
+\sum_{i=1}^d(q^i_h(s,x)-q^i_h(s,y))\partial_iP_{s,t}^hf(y).
\end{split}
\end{equation*}
This suggests that
\begin{equation*}
\|g_1(s,\cdot)\|_{C^\theta_b(\Rnum^d)}
\leq \|q_h(s,\cdot)\|\|P_{s,t}^hf\|_{C^{1+\theta}_b(\Rnum^d)}
+\|q_h(s,\cdot)\|_{C^\theta_b(\Rnum^d)}\|P_{s,t}^hf\|_{C^1_b(\Rnum^d)}.
\end{equation*}
By Lemmas \ref{perturbation} and \ref{contractive}, there exists a constant $K_1 = K_1(d,T,\lambda,C,\theta)>0$ such that when $\|h\|$ is sufficiently small,
\begin{equation*}
\|P_{s,t}^hf\|_{C^{1+\theta}_b(\Rnum^d)} \leq K_1\|f\|_{C^{1+\theta}_b(\Rnum^d)},\;\;\;
\forall\;0\leq s\leq t\leq T.
\end{equation*}
Recall the following interpolation inequality of H\"{o}lder spaces \cite[Section 2.7.2, Theorem 1]{triebel1978interpolation}: there exists a constant $K_2 = K_2(\theta)>0$ such that
\begin{equation*}
\|f\|_{C^{1+\theta}_b(\Rnum^d)}
\leq K_2\|f\|_{C_b(\Rnum^d)}^{\frac{1}{2+\theta}}
\|f\|_{C^{2+\theta}_b(\Rnum^d)}^{\frac{1+\theta}{2+\theta}}
\leq K_2\|f\|_{C^{2+\theta}_b(\Rnum^d)}.
\end{equation*}
The above three inequalities, together with Assumption \ref{ass}(b), show that there exists a constant $K = K(d,T,\theta,\lambda,C)>0$ such that
\begin{equation*}
\|g_1\|_{C^{0,\theta}_b([0,t]\times\Rnum^d)}
\leq 2K\|q_h\|_{C^{0,\theta}_b([0,t]\times\Rnum^d)}\|f\|_{C^{2+\theta}_b(\Rnum^d)}
\leq 2KL\|f\|_{C^{2+\theta}_b(\Rnum^d)}.
\end{equation*}
Similarly, we can prove that
\begin{equation*}
\|g_2\|_{C^{0,\theta}_b([0,t]\times\Rnum^d)}\leq 2KL\|f\|_{C^{2+\theta}_b(\Rnum^d)}.
\end{equation*}
Then (b) follows from the above two inequalities.
\end{proof}

\begin{lemma}\label{PminusP}
For any $f\in C^2_b(\Rnum^d)$ and $0\leq s\leq t\leq T$, when $\|h\|$ is sufficiently small,
\begin{equation*}
P_{s,t}^hf(x)-P_{s,t}f(x) = \int_s^tP_{s,u}(\mathcal{A}_u^h-\mathcal{A}_u)P^h_{u,t}f(x)du.
\end{equation*}
\end{lemma}

\begin{proof}
Since both $X^h$ and $X$ satisfy the regular conditions, it follows from Lemma \ref{kolmogorov} that their transition semigroups satisfy the Kolmogorov backward equations \eqref{backward}. Therefore, the function $u(s,x) = P_{s,t}^hf(x)-P_{s,t}f(x)$ is the bounded classical solution to the following parabolic equation:
\begin{equation}\label{parabolicv}
\begin{cases}
\partial_s u(s,x) = -\mathcal{A}_su(s,x)-h(s)g(s,x),\;\;\;0\leq s<t\\
u(t,x) = 0,
\end{cases}
\end{equation}
where $g(s,x)$ is defined in Theorem \ref{gestimation}. It follows from Ito's formula that
\begin{equation*}
\begin{split}
du(s,X_s) &= [\partial_su(s,X_s)+\mathcal{A}_su(s,X_s)]ds+\nabla u(s,X_s)^T\sigma(s,X_s)dW_s\\
&= -h(s)g(s,X_s)ds+\nabla u(s,X_s)^T\sigma(s,X_s)dW_s.
\end{split}
\end{equation*}
If $|x|<R$ and $s<r<t$, we have
\begin{equation*}
u(s,x) = \Enum_{s,x}u(r\wedge \tau_R,X_{r\wedge \tau_R})
+\Enum_{s,x}\int_s^rh(u)g(u,X_u)I_{\{u\leq \tau_R\}}du.
\end{equation*}
Since $X$ is nonexplosive up to time $T$ and $u\in C_b([0,t]\times\Rnum^d)$, we have
\begin{equation*}
\lim_{r\rightarrow t}\lim_{R\rightarrow\infty}\Enum_{s,x}u(r\wedge \tau_R,X_{r\wedge \tau_R})
= \lim_{r\rightarrow t}\Enum_{s,x}u(r,X_r) = \Enum_{s,x}u(t,X_t) = 0.
\end{equation*}
It follows from Theorem \ref{gestimation} that $g\in C_b([0,t]\times\Rnum^d)$. By the dominated convergence theorem, we finally obtain that
\begin{equation*}
u(s,x) = \int_s^t\Enum_{s,x}h(u)g(u,X_u)du
= \int_s^t\Enum_{s,x}(\mathcal{A}_u^h-\mathcal{A}_u)P_{u,t}^hf(X_u)du,
\end{equation*}
which gives the desired result.
\end{proof}

The following lemma, whose proof can be found in \cite[Theorem 2.7]{lorenzi2011optimal}, gives the Schauder estimate for parabolic equations.

\begin{lemma}\label{Schauder}
Fix $0\leq t\leq T$ and $0<\gamma<1$. For any $f\in C^{2+\gamma}_b(\Rnum^d)$ and $g\in C_b^{0,\gamma}([0,t]\times\Rnum^d)$, the Cauchy problem of the parabolic equation
\begin{equation}\label{drive}
\begin{cases}
\partial_su(s,x) = -\mathcal{A}_su(s,x)-g(s,x),\;\;\;0\leq s\leq t,\\
u(t,x) = f(x),
\end{cases}
\end{equation}
has a unique bounded classical solution. Moreover, there exists a constant $K = K(d,T,\gamma,\lambda,C)>0$ such that
\begin{equation*}
\|u\|_{C_b^{0,2+\gamma}([0,t]\times\Rnum^d)}
\leq K\left[\|f\|_{C_b^{2+\gamma}(\Rnum^d)}+\|g\|_{C_b^{0,\gamma}([0,t]\times\Rnum^d)}\right].
\end{equation*}
\end{lemma}

The above Schauder estimate shows that if the driving term $g$ is of the class $C_b^{0,\gamma}([0,t]\times\Rnum^d)$, then the solution $u$ is of the class $C_b^{0,2+\gamma}([0,t]\times\Rnum^d)$.

\begin{lemma}\label{limit}
For any $f\in C^{2+\theta}_b(\Rnum^d)$, $\phi\in C[0,T]$, and $0\leq s\leq t\leq T$,
\begin{equation*}
\lim_{\epsilon\rightarrow 0}\frac{1}{\epsilon}(P_{0,t}^{\epsilon\phi}f(x)-P_{0,t}f(x))
= \int_0^t \phi(s)P_{0,s}\mathcal{L}_sP_{s,t}f(x)ds.
\end{equation*}
\end{lemma}

\begin{proof}
It follows from Lemma \ref{PminusP} that when $\epsilon$ is sufficiently small,
\begin{equation}\label{difference}
\frac{1}{\epsilon}(P_{0,t}^{\epsilon\phi}f(x)-P_{0,t}f(x))
= \frac{1}{\epsilon}\int_0^tP_{0,s}(\mathcal{A}_s^{\epsilon\phi}-\mathcal{A}_s)P_{s,t}^{\epsilon\phi}f(x)ds
= \int_0^t\phi(s)\Enum_xg_\epsilon(s,X_s)ds,
\end{equation}
where $g_\epsilon(s,x) = \mathcal{L}^{\epsilon\phi}_sP_{s,t}^{\epsilon\phi}f(x)$. By Theorem \ref{gestimation}, we have
\begin{equation}\label{bound}
\|g_\epsilon\|_{C^{0,\theta}_b([0,t]\times\Rnum^d)}\leq 4KL\|f\|_{C^{2+\theta}_b(\Rnum^d)}.
\end{equation}
Thus we obtain that
\begin{equation*}
\begin{split}
\sup_{0\leq s\leq t}\|\epsilon\phi(s)g_\epsilon(s,\cdot)\|_{C^\theta_b(\Rnum^d)}
\leq \epsilon\|\phi\|\|g_\epsilon\|_{C^{0,\theta}_b([0,t]\times\Rnum^d)}\rightarrow 0,\;\;\;\mbox{as}\;\epsilon\rightarrow 0.
\end{split}
\end{equation*}
By \eqref{parabolicv} and Lemma \ref{Schauder}, we have
\begin{equation*}
\sup_{0\leq s\leq t}\|P_{s,t}^{\epsilon\phi}f-P_{s,t}f\|_{C^{2+\theta}_b(\Rnum^d)}\rightarrow 0,\;\;\;\mbox{as}\;\epsilon\rightarrow 0.
\end{equation*}
This shows that as $\epsilon\rightarrow 0$,
\begin{equation}\label{convergence}
\begin{split}
g_\epsilon(s,x) &= \sum_{i=1}^dq^i_{\epsilon\phi}(s,x)\partial_iP^{\epsilon\phi}_{s,t}f(x)
+\frac{1}{2}\sum_{i,j=1}^dr^{ij}_{\epsilon\phi}(s,x)\partial_{ij}P^{\epsilon\phi}_{s,t}f(x)\\
&\rightarrow \sum_{i=1}^dq^i(s,x)\partial_iP_{s,t}f(x)
+\frac{1}{2}\sum_{i,j=1}^dr^{ij}(s,x)\partial_{ij}P_{s,t}f(x)
= \mathcal{L}_sP_{s,t}f(x).
\end{split}
\end{equation}
Thus it follows from \eqref{difference}, \eqref{bound}, \eqref{convergence}, and the dominated convergence theorem that
\begin{equation*}
\lim_{\epsilon\rightarrow 0}\frac{1}{\epsilon}(P_{0,t}^{\epsilon\phi}f(x)-P_{0,t}f(x))
= \int_0^t\phi(s)\Enum_x\lim_{\epsilon\rightarrow 0}g_\epsilon(s,X_s)ds
= \int_0^t\phi(s)\Enum_x\mathcal{L}_sP_{s,t}f(X_s)ds,
\end{equation*}
which gives the desired result.
\end{proof}

The following lemma gives the regularity of the probability densities for $X$.

\begin{lemma}\label{secondorder}
If the regular conditions (a) and (b) are satisfied, then $X_t$ has a positive probability density $p_t\in H^2_{\loc}(\Rnum^d)$ with respect to the Lebesgue measure for almost all $0\leq t\leq T$. In particular, any stationary distribution of $X$, whenever it exists, must have a positive probability density $\mu\in H^2_{\loc}(\Rnum^d)$.
\end{lemma}

\begin{proof}
For any $\phi\in C_c^\infty((0,T)\times\Rnum^d)$, it follows from Ito's formula that
\begin{equation}\label{Ito}
d\phi(t,X_t) = [\partial_t\phi(t,X_t)+\mathcal{A}_t\phi(t,X_t)]dt+\nabla\phi(t,X_t)\sigma(t,X_t)dW_t.
\end{equation}
Since $\nabla\phi^Ta\nabla\phi$ is bounded, we have
\begin{equation}\label{backwardeq}
\int_0^T\int_{\Rnum^d}[\partial_t\phi(t,x)+\mathcal{A}_t\phi(t,x)]p_t(dx)dt = 0,
\end{equation}
where $p_t(dx)$ is the probability distribution of $X_t$. Under the regular conditions (a) and (b), it follows from \eqref{backwardeq} and \cite[Corollary 6.4.3 and Theorem 6.2.7]{bogachev2015fokker} that there exists a positive function $\rho\in C((0,T)\times\Rnum^d)$ satisfying $\rho(t,\cdot)\in H^1_{\loc}(\Rnum^d)$ for any $0<t<T$ such that
\begin{equation*}
p_t(dx)dt = \rho(t,x)dtdx.
\end{equation*}
This shows that $X_t$ has a positive probability density $p_t = \rho(t,\cdot)$ for almost all $0\leq t\leq T$. Moreover, it follows from \eqref{Ito} that for any $\epsilon>0$ and $\epsilon<t<T-\epsilon$,
\begin{equation*}
\int_{\Rnum^d}\phi(t,x)p_t(x)dx-\int_{\Rnum^d}\phi(\epsilon,x)p_\epsilon(x)dx-
\int_{\epsilon}^t\int_{\Rnum^d}[\partial_s\phi(s,x)+\mathcal{A}_s\phi(s,x)]p_s(x)dxds = 0.
\end{equation*}
Under the regular condition (a), it follows from the weak differentiability of $p_t$ and the integration by parts formula that
\begin{equation}\label{integralparts}
\begin{split}
&\int_{\Rnum^d}p_t(x)\phi(t,x)dx-\int_{\Rnum^d}p_\epsilon(x)\phi(\epsilon,x)dx
-\int_{\epsilon}^t\int_{\Rnum^d}p_s(x)\partial_s\phi(s,x)dxds \\
&+\int_{\epsilon}^t\int_{\Rnum^d}\sum_{i,j=1}^d\left[\frac{1}{2}a^{ij}\partial_ip_s\partial_j\phi
+\left(b^i-\frac{1}{2}\partial_ja^{ij}\right)\partial_ip_s\phi
+\left(\partial_ib^i-\frac{1}{2}\partial_{ij}a^{ij}\right)p_s\phi\right]dxds = 0.
\end{split}
\end{equation}
For each $R>0$, let $B_R = \{x\in\Rnum^d:|x|<R\}$ and let $\zeta_R\in C_c^{\infty}(\Rnum^d)$ be a cutoff function satisfying $0\leq\zeta_R\leq 1$, $\zeta_R(x) = 1$ for $|x|\leq R/4$, and $\zeta_R(x) = 0$ for $|x|>3R/4$. We then defined a function $\rho_R:(0,T)\times\Rnum^d\rightarrow\Rnum$ by
\begin{equation*}
\rho_R(t,x) = \rho(t,x)\zeta_R(x).
\end{equation*}
Since \eqref{integralparts} holds for all $\phi\in C_c^\infty((0,T)\times\Rnum^d)$, replacing $\phi(t,x)$ by $\phi(t,x)\zeta_R(x)$ in \eqref{integralparts} implies that for any $\phi\in C_c^{\infty}((\epsilon, T-\epsilon)\times B_R)$ and $\epsilon<t<T-\epsilon$,
\begin{equation}\label{parts}
\begin{split}
&\int_{B_R}\rho_R(t,x)\phi(t,x)dx-\int_{B_R}\rho_R(\epsilon,x)\phi(\epsilon,x)dx
-\int_\epsilon^t\int_{B_R}\rho_R(s,x)\partial_s\phi(s,x)dxds \\
&+\int_\epsilon^t\int_{B_R}\sum_{i,j=1}^d
\left[\frac{1}{2}a^{ij}\partial_i\rho_R\partial_j\phi
+\left(b^i-\frac{1}{2}\partial_ja^{ij}\right)\partial_i\rho_R\phi
+\left(\partial_ib^i-\frac{1}{2}\partial_{ij}a^{ij}\right)\rho_R\phi\right]dxds \\
&= -\int_\epsilon^t\int_{B_R}g_R(s,x)\phi(s,x)dxds,
\end{split}
\end{equation}
where
\begin{equation*}
g_R = \sum_{i,j=1}^d\left[\frac{1}{2}\partial_ja^{ij}\rho\partial_{i}\zeta_R
+\frac{1}{2}a^{ij}\rho\partial_{ij}\zeta_R
-\left(\partial_ib^i-\frac{1}{2}\partial_{ij}a^{ij}\right)\rho\partial_i\zeta_R\right].
\end{equation*}
For convenience, we define a family of second-order elliptic operator $\{\mathcal{B}_t:t\geq 0\}$ as
\begin{equation*}
\mathcal{B}_tu = \frac{1}{2}\sum_{i,j=1}^d\partial_j(a^{ij}\partial_iu)
-\sum_{i,j=1}^d\left(b^i-\frac{1}{2}\partial_ja^{ij}\right)\partial_iu
-\sum_{i,j=1}^d\left(\partial_ib^i-\frac{1}{2}\partial_{ij}a^{ij}\right)u.
\end{equation*}
It follows from \eqref{parts} that $\rho_R$ is the weak solution of the initial and boundary value problem of the parabolic equation \cite[Chapter VI, Section 1]{lieberman1996second}
\begin{equation*}
\begin{cases}
\partial_tu-\mathcal{B}_tu = g_R &\textrm{in\;} (\epsilon,T-\epsilon)\times B_R,\\
u = 0 &\textrm{on\;} (\epsilon, T-\epsilon)\times\partial B_R,\\
u = \rho_R &\textrm{on\;} \{\epsilon\}\times B_R.
\end{cases}
\end{equation*}
Since $\rho(\epsilon,\cdot)\in H^1(B_R)$, the initial value of this parabolic equation satisfies $\rho_R(\epsilon,\cdot)\in H^1(B_R)$. Thus it follows from \cite[Theorem 6.6]{lieberman1996second} that $\rho_R(t,\cdot)\in H^2(B_R)$ for almost all $\epsilon< t<T-\epsilon$. By the arbitrariness of $R$ and $\epsilon$,  we conclude that $p_t\in H^2_{\loc}(\Rnum^d)$ for almost all $0\leq t\leq T$.
\end{proof}

We shall now state the main results of this section. Recall that we always assume the regular conditions (a)-(e) and Assumption \ref{ass} to be satisfied. For any $t>0$, let $p_t$ denote the probability density of $X_t$. The following theorem gives an explicit expression for the response function of an observable.

\begin{theorem}\label{expression}
Let $f\in C^{2+\theta}_b(\Rnum^d)$ be an observable and let $R_f$ be the response function of $f$. Then for any $0\leq s\leq t\leq T$,
\begin{equation*}
R_f(s,t) = \Enum\mathcal{L}_sP_{s,t}f(X_s) = \int_{\Rnum^d}\mathcal{L}_sP_{s,t}f(x)p_s(x)dx.
\end{equation*}
\end{theorem}

\begin{proof}
For any $0\leq t\leq T$, it is easy to see that
\begin{equation*}
\Enum f(X_t) = \int_{\Rnum^d}\Enum_xf(X_t)p_0(dx) = \int_{\Rnum^d}P_{0,t}f(x)p_0(dx) = \Enum P_{0,t}f(X_0).
\end{equation*}
Thus it follows from Lemma \ref{limit} that for any $0\leq t\leq T$ and $\phi\in C_c^{\infty}(0,t)$,
\begin{equation*}
\begin{split}
\langle\frac{\delta F_t}{\delta h}\Big|_{h=0},\phi\rangle
&= \lim_{\epsilon\rightarrow 0}\frac{1}{\epsilon}(F_t(\epsilon\phi)-F_t(0))
= \lim_{\epsilon\rightarrow 0}\frac{1}{\epsilon}(\Enum f(X^{\epsilon\phi}_t)-\Enum f(X_t))\\
&= \lim_{\epsilon\rightarrow 0}\Enum\left[\frac{1}{\epsilon}(P^{\epsilon\phi}_{0,t}f(X_0)-P_{0,t}f(X_0))\right]\\
&= \int_0^t\phi(s)\Enum P_{0,s}\mathcal{L}_sP_{s,t}f(X_0)ds\\
&= \int_0^t\phi(s)\Enum\mathcal{L}_sP_{s,t}f(X_s)ds,
\end{split}
\end{equation*}
which gives the desired result.
\end{proof}

We are now in a position to prove the Agarwal-type FDT. Here we still assume the regular conditions (a)-(e) and Assumption \ref{ass} to be satisfied.

\begin{theorem}\label{Agarwal1}
Fix $0\leq s\leq t\leq T$ such that $X_s$ has a positive probability density $p_s\in H^2_{\loc}(\Rnum^d)$. Assume that $q(s,\cdot)\in C_c^1(\Rnum^d)$ and $r(s,\cdot)\in C_c^2(\Rnum^d)$. Let $v_s$ be a function on $\Rnum^d$ defined by
\begin{equation*}
v_s(x) = \frac{\mathcal{L}_s^{*}p_s(x)}{p_s(x)},
\end{equation*}
where
\begin{equation*}
\mathcal{L}^*_sf(x) = -\sum_{i=1}^d\partial_i(q^i(s,x)f(x))
+\frac{1}{2}\sum_{i,j=1}^d\partial_{ij}(r^{ij}(s,x)f(x)),\;\;\;f\in W^{2,1}_{\loc}(\Rnum^d)
\end{equation*}
is the adjoint operator of $\mathcal{L}_s$. Then for any $f\in C^{2+\theta}_b(\Rnum^d)$,
\begin{equation*}
R_f(s,t) = \Enum f(X_t)v_s(X_s).
\end{equation*}
\end{theorem}

\begin{proof}
For any $0\leq s\leq t$ and any measurable function $u$ on $\Rnum^d$ such that $\Enum|u(X_s)|<\infty$, we have
\begin{equation}\label{formula}
\Enum f(X_t)u(X_s) = \Enum u(X_s)\Enum\{f(X_t)|X_s\} = \Enum P_{s,t}f(X_s)u(X_s).
\end{equation}
Since $q(s,\cdot)\in C_c^1(\Rnum^d), r(s,\cdot)\in C_c^2(\Rnum^d)$, and $p_s\in W^{2,1}_{\loc}(\Rnum^d)$, we have $q(s,\cdot)p_s\in W^{1,1}_{\loc}(\Rnum^d)$ and $r(s,\cdot)p_s\in W^{2,1}_{\loc}(\Rnum^d)$ \cite[Section 5.2.3, Theorem 1]{evans2010partial}. Thus it follows from Theorem \ref{expression} and the integration by parts formula that
\begin{equation*}
\begin{split}
R_f(s,t) &= \int_{\Rnum^d}\mathcal{L}_sP_{s,t}f(x)p_s(x)dx\\
&= \sum_{i=1}^d\int_{\Rnum^d}q^i(s,x)p_s(x)\partial_iP_{s,t}f(x)dx
+\frac{1}{2}\sum_{i,j=1}^d\int_{\Rnum^d}r^{ij}(s,x)p_s(x)\partial_{ij}P_{s,t}f(x)dx\\
&= -\sum_{i=1}^d\int_{\Rnum^d}\partial_i(q^i(s,x)p_s(x))P_{s,t}f(x)dx
+\frac{1}{2}\sum_{i,j=1}^d\int_{\Rnum^d}\partial_{ij}(r^{ij}(s,x)p_s(x))P_{s,t}f(x)dx\\
&= \int_{\Rnum^d}\mathcal{L}^*_sp_s(x)P_{s,t}f(x)dx
= \int_{\Rnum^d}P_{s,t}f(x)v_s(x)p_s(x)dx = \Enum P_{s,t}f(X_s)v_s(X_s),
\end{split}
\end{equation*}
which gives the desired result.
\end{proof}

By Lemma \ref{secondorder}, $X_t$ has a positive probability density $p_s\in H^2_{\loc}(\Rnum^d)$ for almost all $0\leq s\leq T$. Therefore, the first condition in the above theorem is automatically satisfied for almost all $s$.

\begin{remark}
The above theorem indicates that for inhomogeneous diffusion processes, the response of an observable $f$ to a small external perturbation can be represented as the correlation function of this observable and the conjugate observable $v_s = \mathcal{L}^*_sp_s/p_s$, which generally depends on the early time $s$. If $X$ is inhomogeneous, then the generator $\mathcal{L}_s$ will depend on $s$. If $X$ is nonstationary, then the probability distribution $p_s$ will depends on $s$. If we hope the conjugate observable $v_s$ to be independent of $s$, the diffusion process must be both homogenous and stationary.
\end{remark}

The following theorem shows that the conjugate observable in the Aargwal-type FDT is unique.

\begin{theorem}\label{uniqueness1}
Fix $0\leq s< T$ such that $X_s$ has a positive probability density $p_s$. Assume that there exists another function $\tilde v_s\in L^1(p_s)$ on $\Rnum^d$ such that
\begin{equation*}
\Enum f(X_t)v_s(X_s) = \Enum f(X_t)\tilde{v}_s(X_s),\;\;\;
\forall\;f\in C_c^{\infty}(\Rnum^d), s<t\leq T.
\end{equation*}
Then $v_s = \tilde v_s$ almost everywhere.
\end{theorem}

\begin{proof}
From \eqref{formula}, it is easy to check that
\begin{equation*}
\int_{\Rnum^d}P_{s,t}f(x)(v_s(x)-\tilde v_s(x))p_s(x)dx = 0.
\end{equation*}
Since $f\in C_c^{\infty}(\Rnum^d)$, taking $t\rightarrow s$ in the above equation and applying the dominated convergence theorem give rise to
\begin{equation*}
\int_{\Rnum^d}f(x)(v_s(x)-\tilde v_s(x))p_s(x)dx = 0,
\end{equation*}
By the arbitrariness of $f$, we obtain the desired result.
\end{proof}

Let us recall the following important definition from stochastic thermodynamics.

\begin{definition}
Let $X$ be homogenous and stationary. Then $X$ is said to be in an \emph{equilibrium state} if $X$ is symmetric with respect to its stationary distribution and $X$ is said to be in an \emph{NESS} if $X$ is nonsymmetric with respect to its stationary distribution.
\end{definition}

If $X$ is homogeneous and stationary, then the generator $\mathcal{A}$, the operator $\mathcal{L}$, and the functions $b$, $a$, $q$, and $r$ are all independent of the time variable $t$. Moreover, the transition semigroup $P_{s,t}$ only depends on the time difference $t-s$ and can be formally represented as $P_{s,t} = e^{\mathcal{A}(t-s)}$. In this case, the above result reduces to the Agarwal-type FDT for an NESS.

\begin{theorem}\label{Agarwal2}
Let $X$ be homogeneous and stationary with $\mu\in H^2_{\loc}(\Rnum^d)$ being the positive stationary density. Assume that $q\in C_c^1(\Rnum^d)$ and $r\in C_c^2(\Rnum^d)$. Let $v$ be a function on $\Rnum^d$ defined by
\begin{equation*}
v(x) = \frac{\mathcal{L}^{*}\mu(x)}{\mu(x)}.
\end{equation*}
Then for any $f\in C^{2+\theta}_b(\Rnum^d)$ and $0\leq s\leq t$,
\begin{equation*}
R_f(s,t) = \Enum f(X_t)v(X_s) = \int_{\Rnum^d}\mathcal{L}e^{\mathcal{A}(t-s)}f(x)\mu(dx).
\end{equation*}
\end{theorem}

\begin{proof}
The desired result follows directly from Theorems \ref{expression} and \ref{Agarwal1}.
\end{proof}

According to the above theorem, whenever $X$ is homogenous and stationary, whether in an equilibrium state or in an NESS, the conjugate observable $v$ does not depend on the early time $s$ and the response function $R_f(s,t)$ only depends on the time difference $t-s$.

Lemma \ref{secondorder} shows that the stationary distribution of $X$, if it exists, must have a positive probability density $\mu\in H^{2}_{\loc}(\Rnum^d)$. Theoretical physicists may be particularly interested in the following proposition, which contains very weak conditions for the higher-order regularity of the stationary density and generalizes the classical results on NESS \cite[Theorem 3.2.5]{jiang2004mathematical} to a large extent. We do not assume the regular conditions (a)-(e) in the following proposition.

\begin{proposition}\label{proregularityinvariant}
Let $X$ be homogenous. Assume that $a$ satisfies the following locally elliptic condition: there exists a positive function $\eta:\Rnum^d\rightarrow\Rnum$ such that
\begin{equation}\label{localellipitic}
\xi^Ta(x)\xi\geq \eta(x)|\xi|^2.
\end{equation}
Then the following five statements hold:
\begin{itemize}\style
\item[(a)] If $b$ is locally bounded and $a\in W^{1,p}_{\loc}(\Rnum^d)$ for some $p>d$, then any stationary distribution of $X$, if it exists, has a positive probability density $\mu\in C^{1-d/p}_{\loc}(\Rnum^d)\cap W^{1,p}_{\loc}(\Rnum^d)$.
\item[(b)] If $b,a\in C^\alpha_{\loc}(\Rnum^d)$ for some $0<\alpha<1$, then the stationary distribution of $X$  must be unique.
\item[(c)] If $b\in C^1(\Rnum^d)$ and $a\in C^2(\Rnum^d)$, then $\mu\in H^2_{\loc}(\Rnum^d)$.
\item[(d)] If $b\in C^m(\Rnum^d)$ and $a\in C^{m+1}(\Rnum^d)$ for some integer $m\geq 2$, then $\mu\in H^m_{\loc}(\Rnum^d)$.
\item[(e)] If $b,a\in C^\infty(\Rnum^d)$, then $\mu\in C^\infty(\Rnum^d)$.
\end{itemize}
\end{proposition}

\begin{proof}
The first part of (a) follows from \cite[Corollaries 1.6.9 and 1.7.2]{bogachev2015fokker} and the Sobolev embedding theorem, which claims that $W^{1,p}(U)$ with $p>d$ can be embedded into $C^{1-d/p}_b(\bar U)$ for any open ball $U\subset\Rnum^d$. The second part of (a) follows from Ito's formula and \cite[Theorem 5.3.3]{bogachev2015fokker}. Moreover, (b) follows from \cite[Theorem 8.1.15]{lorenzi2006analytical}.

We next prove (c). For any $\phi\in C_c^2(\Rnum^d)$, it follows from Ito's formula that
\begin{equation*}
d\phi(X_t) = \mathcal{A}\phi(X_t)dt+\nabla\phi(X_t)^T\sigma(X_t)dW_t.
\end{equation*}
Since $\nabla\phi^Ta\nabla\phi$ is bounded, we have
\begin{equation*}
\int_0^tds\int_{\Rnum^d}\mathcal{A}\phi(x)\mu(dx) = \int_0^t\Enum_\mu\mathcal{A}\phi(X_s)ds
= \Enum_\mu\phi(X_t)-\Enum_\mu\phi(X_0) = 0,
\end{equation*}
which suggests that
\begin{equation*}
\int_{\Rnum^d}\mathcal{A}\phi(x)\mu(x)dx = 0.
\end{equation*}
Since $b\in C^1(\Rnum^d)$ and $a\in C^2(\Rnum^d)$, it follows from (a) that $\mu\in H^1_{\loc}(\Rnum^d)$. By the integration by parts formula, it is easy to check that
\begin{equation*}
\int_{\Rnum^d}\sum_{i,j=1}^d\left[\frac{1}{2}a^{ij}\partial_i\mu\partial_j\phi
+\left(b^i-\frac{1}{2}\partial_ja^{ij}\right)\partial_i\mu\phi
+\left(\partial_ib^i-\frac{1}{2}\partial_{ij}a^{ij}\right)\mu\phi\right]dx = 0.
\end{equation*}
For any bounded open subsets $U\subset\Rnum^d$, since $C_c^2(U)$ is dense in $H^1_0(U)$, it is easy to check that the above equality holds for any $\phi\in H^1_0(U)$. This suggests that $\mu$ is a classical weak solution \cite[Section 6.1]{evans2010partial} for the following elliptic equation of the divergence form:
\begin{equation*}
-\frac{1}{2}\sum_{i,j=1}^d\partial_j(a^{ij}\partial_i\mu)
+\sum_{i,j=1}^d\left(b^i-\frac{1}{2}\partial_ja^{ij}\right)\partial_i\mu
+\sum_{i,j=1}^d\left(\partial_ib^i-\frac{1}{2}\partial_{ij}a^{ij}\right)\mu = 0\;\;\;\mbox{in}\;U.
\end{equation*}
By \cite[Section 6.3.1, Theorem 1]{evans2010partial}, the weak solution must satisfy $\mu\in H^2_{\loc}(U)$. By the arbitrariness of the bounded open subsets $U$, we have $\mu\in H^2_{\loc}(\Rnum^d)$.

Similarly, if $b\in C^m(\Rnum^d)$ and $a\in C^{m+1}(\Rnum^d)$ for some integer $m\geq 2$, the weak solution must satisfy $\mu\in H^m_{\loc}(U)$ \cite[Section 6.3.1, Theorem 2]{evans2010partial}. Finally, if $b,a\in C^\infty(\Rnum^d)$, the weak solution must satisfy $\mu\in C^\infty(\Rnum^d)$ \cite[Section 6.3.1, Theorem 3]{evans2010partial}. Thus we have proved (d) and (e).
\end{proof}

\section{The Seifert-Speck-type FDT}
The Seifert-Speck-type FDT only holds for an NESS and cannot be extended to general nonequilibrium states. Therefore, we always assume that $X$ is homogeneous and stationary in this section and we shall use the semigroup theory of homogenous Markov processes to study this type of FDT.  Here we still assume the regular conditions (a)-(e) and Assumption \ref{ass} to be satisfied. If $X$ is homogenous, the transition semigroup $\{P_t\}$ of $X$ is defined as
\begin{equation*}
P_tf(x) = \Enum_xf(X_t) := \Enum\{f(X_t)|X_0=x\},\;\;\;f\in B(\Rnum^d).
\end{equation*}

\begin{lemma}
If the regular conditions (a) and (b) are satisfied, then $\{P_t\}$ is a contractive semigroup on $C_b(\Rnum^d)$.
\end{lemma}

\begin{proof}
By \cite[Corollary 4.7]{metafune2002feller}, the regular conditions (a) and (b) imply that $X$ is strong Feller, that is, $P_tf\in C_b(\Rnum^d)$ for any bounded measurable function $f$ and $t>0$. Therefore, $P_t$ is a bounded linear operator on $C_b(\Rnum^d)$. The semigroup property and contractive property are obvious.
\end{proof}

We stress here that if $b$ or $a$ is unbounded, then $\{P_t\}$ may not be a strongly continuous semigroup on $C_b(\Rnum^d)$ and thus the classical semigroup theory is not applicable. Even for the OU process, $\{P_t\}$ is not strongly continuous on $C_b(\Rnum^d)$ \cite{daprato1995ornstein} and we cannot define the generator in the usual sense. Fortunately, $\{P_t\}$ is a weakly continuous semigroup and we can define the generator in the weak sense. To make the paper self-contained, we recall the definition of a weakly continuous semigroup and its weak generator as follows \cite[Section 2.3]{lorenzi2006analytical}.

\begin{definition}
The contraction semigroup $\{P_t\}$ is called \emph{weakly continuous} on $C_b(\Rnum^d)$ if the following two conditions are satisfied:
\begin{itemize}\style
\item[(a)] For any $f\in C_b(\Rnum^d)$ and $x\in\Rnum^d$, $P_tf(x)$ is continuous with respect to $t$.
\item[(b)] For any sequence $\{f_n\}$ in $C_b(\Rnum^d)$, if $\{f_n\}$ is uniformly bounded and converges pointwise to $f\in C_b(\Rnum^d)$, then $\{P_tf_n\}$ converges pointwise to $P_tf$.
\end{itemize}
\end{definition}

By the dominated convergence theorem, it is easy to check that $\{P_t\}$ is a weakly continuous semigroup on $C_b(\Rnum^d)$.

\begin{definition}
For any $f\in C_b(\Rnum^d)$, we say that $f\in D(\mathcal{A})$ if
\begin{equation}\label{bounded}
\sup_{\epsilon>0}\left\|\frac{1}{\epsilon}(P_\epsilon f-f)\right\|<\infty
\end{equation}
and there exists $u\in C_b(\Rnum^d)$ such that as $\epsilon\downarrow0$,
\begin{equation*}
\frac{1}{\epsilon}(P_\epsilon f(x)-f(x))\rightarrow u(x),\;\;\;\forall\;x\in\Rnum^d.
\end{equation*}
Let $\mathcal{A}: D(\mathcal{A})\rightarrow C_b(\Rnum^d)$ be a linear operator defined by
\begin{equation*}
\mathcal{A}f(x) = \lim_{\epsilon\rightarrow 0}\frac{1}{\epsilon}(P_\epsilon f(x)-f(x)).
\end{equation*}
Then $\mathcal{A}$ is called the \emph{weak generator} of the semigroup $\{P_t\}$.
\end{definition}

In fact, it can be proved that $D(\mathcal{A})$ is a dense subset of $C_b(\Rnum^d)$ in the following sense: for any $f\in C_b(\Rnum^d)$, there exists a sequence $\{f_n\}\subset D(\mathcal{A})$ such that $f_n\rightarrow f$ uniformly on every compact subsets of $\Rnum^d$ \cite[Proposition 2.3.5]{lorenzi2006analytical}. Specifically, under the regular conditions, $D(\mathcal{A})$ can be characterized explicitly as follows \cite[Propositions 2.3.6 and 4.1.1]{lorenzi2006analytical}:
\begin{equation*}
D(\mathcal{A}) = \{u\in C_b(\Rnum^d)\cap\underset{1\leq p<\infty}{\textstyle\bigcap}H_{\loc}^p(\Rnum^d): \mathcal{A}u\in C_b(\Rnum^d)\}.
\end{equation*}
From the above characterization, it is easy to see that $C_c^2(\Rnum^d)\subset D(\mathcal{A})$. Moreover, if $b$ and $a$ are bounded, then $C_b^2(\Rnum^d)\subset D(\mathcal{A})$.

Since $X$ is homogeneous, the trivariate functions defined in \eqref{trivariate} do not depend on the time variable $t$ and reduce to
\begin{equation*}
\bar b:[-1,1]\times\Rnum^d\rightarrow \Rnum^d,\;\;\;
\bar a:[-1,1]\times\Rnum^d\rightarrow M_{d\times d}(\Rnum).
\end{equation*}
For any $-1\leq h\leq 1$, we define an \emph{auxiliary diffusion process} $\bar X^h = \{\bar X^h_t:t\geq 0\}$ with perturbed drift $\bar b_h(x) = \bar b(h,x)$ and diffusion matrix $\bar a_h(x) = \bar a(h,x)$. We stress here that the perturbed process $X^h$ and the auxiliary process $\bar X^h$ are different. The perturbed process $X^h$ is an inhomogeneous diffusion process with
\begin{equation*}
\mathcal{A}^h_tf = \sum_{i=1}^d\bar b^i(h(t),x)\partial_if+\sum_{i,j=1}^d\bar a^{ij}(h(t),x)\partial_{ij}f,\;\;\;f\in W^{2,1}_{\loc}(\Rnum^d),
\end{equation*}
where $h = h(t)$ is a continuous function. However, the auxiliary process is a homogenous diffusion process with
\begin{equation*}
\mathcal{\bar A}^hf = \sum_{i=1}^d\bar b^i(h,x)\partial_if+\sum_{i,j=1}^d\bar a^{ij}(h,x)\partial_{ij}f,\;\;\;f\in W^{2,1}_{\loc}(\Rnum^d).
\end{equation*}
where $h$ is taken as a constant. In the following, we do not distinguish $b_h(x)$ and $\bar b(h,x)$ and do not distinguish $a_h(x)$ and $\bar a(h,x)$. The notation should be clear from the context. It is easy to see that when $h$ is sufficiently small, the auxiliary process $\bar X^h$ also satisfies the regular conditions.

\begin{assumption}\label{differentiable}
In the following, we assume that the auxiliary diffusion process $\bar X^h$ satisfies the following two conditions.
\begin{itemize}\style
\item [(a)] When $h$ is sufficiently small, $\bar X^h$ has a stationary density $\mu_h$. The stationary density of $X$ is denoted by $\mu$.
\item [(b)] The stationary density $\mu_h$ is differentiable in $L^1(\Rnum^d)$ at $h = 0$. In other words, there exists $\nu\in L^1(\Rnum^d)$ such that
    \begin{equation*}
    \frac{1}{h}(\mu_h-\mu)\xrightarrow{L^1(\Rnum^d)}\nu,\;\;\;\mbox{as}\;h\rightarrow 0.
    \end{equation*}
\end{itemize}
\end{assumption}

There are many verifiable conditions that can guarantee Assumption \ref{differentiable}(a)-(b). For example, it is widely known \cite[Corollary 2.4.2]{bogachev2015fokker} that if $a$ and $b$ are locally bounded and if there exist a constant $C>0$ and a Lyapunov function $V\in C^2(\Rnum^d)$ satisfying $V(x)\rightarrow \infty$ as $|x|\rightarrow \infty$ such that
\begin{equation*}
\mathcal{\bar A}^hV(x) \leq -C
\end{equation*}
outside a compact set, then $\bar{X}^h$ has a stationary distribution. Moreover, we have the following lemma, whose proof can be found in \cite{bogachev2016differentiability}.

\begin{lemma}\label{criterion}
Assume that $b_h$, $a_h$, and $\partial_ia_h$ are continuously differentiable with respect to $h$ for any $1\leq i\leq d$. Assume that $a_h$, $a_h^{-1}$, and $\partial_ha_h$ are uniformly bounded and there exist two constants $C,k>0$ such that
\begin{equation*}
|\partial_ia_h(x)|+|\partial_h\partial_ia_h(x)|+|b_h(x)|+|\partial_hb_h(x)|\leq C(1+|x|^k),\;\;\;\forall\;1\leq i\leq d,h\in[-1,1],x\in\Rnum^d.
\end{equation*}
Assume also that
\begin{equation}
\lim_{|x|\rightarrow \infty}\sup_{h\in[-1,1]}b_h(x)^Tx = -\infty.
\end{equation}
Then Assumption \ref{differentiable}(b) holds.
\end{lemma}

The relationship between the weak generators of the original and auxiliary processes can be seen from the following lemma.

\begin{lemma}\label{uniform}
If $f\in D(\mathcal{A})\cap C^2_b(\Rnum^d)$, then $f\in D(\mathcal{\bar A}^h)$ when $h$ is sufficiently small. In this case, we have
\begin{equation*}
\mathcal{\bar A}^hf = h\mathcal{L}^hf+\mathcal{A}f.
\end{equation*}
\end{lemma}

\begin{proof}
Let $\{\bar P^h_t\}$ denote the transition semigroup of the auxiliary process $\bar X^h$. Then we have
\begin{equation*}
\frac{1}{\epsilon}(\bar P^h_\epsilon f(x)-f(x)) = \frac{1}{\epsilon}(\bar P^h_\epsilon f(x)-P_\epsilon f(x))+\frac{1}{\epsilon}(P_\epsilon f(x)-f(x)).
\end{equation*}
When $h$ sufficiently small, it follows from Lemma \ref{PminusP} and Theorem \ref{gestimation} that
\begin{equation*}
\bar P^h_\epsilon f(x)-P_\epsilon f(x)
= \int_0^\epsilon P_{\epsilon-s}(\mathcal{A}^h-\mathcal{A})\bar P^h_sf(x)ds
= h\int_0^\epsilon\Enum_xg(s,X_{\epsilon-s})ds.
\end{equation*}
where $g(s,x) = \mathcal{L}^h\bar P^h_sf(x)\in C_b([0,T]\times\Rnum^d)$.
Since $X$ has continuous trajectories, it follows from the dominated convergence theorem that $\Enum_xg(s,X_{\epsilon-s})$ as a function of $s$ and $\epsilon$ is continuous on $0\leq s\leq\epsilon\leq T$. Thus we obtain that
\begin{equation*}
\lim_{\epsilon\rightarrow 0}\frac{1}{\epsilon}(\bar P^h_\epsilon f(x)-P_\epsilon f(x))
= h\lim_{\epsilon\rightarrow 0}\frac{1}{\epsilon}\int_0^\epsilon\Enum_xg(s,X_{\epsilon-s})ds
= h\Enum_xg(0,X_0) = h\mathcal{L}^hf(x).
\end{equation*}
Moreover, it follows from Theorem \ref{gestimation} that
\begin{equation*}
\sup_{0<\epsilon\leq T}\left\|\frac{1}{\epsilon}(\bar P^h_\epsilon f-P_\epsilon f)\right\|
\leq |h|\|g\|_{C_b([0,T]\times\Rnum^d)}\leq 2KL|h|\|f\|_{C^2_b(\Rnum^d)}.
\end{equation*}
This implies that $f\in D(\mathcal{\bar A}^h)$ and $\mathcal{\bar A}^hf = h\mathcal{L}^hf+\mathcal{A}f$.
\end{proof}

The following lemma shows that formally, the stationary density $\mu$ satisfies $\mathcal{A}^*\mu = 0$.
\begin{lemma}\label{stationary}
For any $f\in D(\mathcal{A})$,
\begin{equation*}
\int_{\Rnum^d}\mu(x)\mathcal{A}f(x)dx = 0.
\end{equation*}
\end{lemma}

\begin{proof}
By the dominated convergence theorem, we have
\begin{equation*}
\begin{split}
\int_{\Rnum^d}\mu(x)\mathcal{A}f(x)dx
&= \int_{\Rnum^d}\mu(x)\lim_{\epsilon\rightarrow 0}\frac{1}{\epsilon}(P_\epsilon f(x)-f(x))dx\\
&= \lim_{\epsilon\rightarrow 0}\frac{1}{\epsilon}\int_{\Rnum^d}\mu(x)(P_\epsilon f(x)-f(x))dx\\
&= \lim_{\epsilon\rightarrow 0}\frac{1}{\epsilon}(\Enum f(X_\epsilon)-\Enum f(X_0)).
\end{split}
\end{equation*}
The fact that $X$ is stationary gives the desired result.
\end{proof}

The following lemma shows that the weak generator $\mathcal{A}$ and the operator $\mathcal{L}$ are formally related by $\mathcal{A}^*\nu = -\mathcal{L}^*\mu$.

\begin{lemma}\label{relation}
For any $f\in D(\mathcal{A})\cap C^2_b(\Rnum^d)$,
\begin{equation*}
\int_{\Rnum^d}\nu(x)\mathcal{A}f(x)dx = -\int_{\Rnum^d}\mu(x)\mathcal{L}f(x)dx,
\end{equation*}
where $\nu$ is the function introduced in Assumption \ref{differentiable}(b).
\end{lemma}

\begin{proof}
It follows from Lemmas \ref{uniform} and \ref{stationary} that
\begin{equation*}
\int_{\Rnum^d}\mu(x)\mathcal{A}f(x)dx = \int_{\Rnum^d}\mu_h(x)\mathcal{\bar A}^hf(x)dx = 0.
\end{equation*}
This fact, together with Assumption \ref{differentiable} and Lemma \ref{uniform}, shows that
\begin{equation*}
\begin{split}
\int_{\Rnum^d}\nu(x)\mathcal{A}f(x)dx
&= \lim_{h\rightarrow 0}\frac{1}{h}\int_{\Rnum^d}(\mu_h(x)-\mu(x))\mathcal{A}f(x)dx
= \lim_{h\rightarrow 0}\frac{1}{h}\int_{\Rnum^d}\mu_h(x)\mathcal{A}f(x)dx\\
&= -\lim_{h\rightarrow 0}\frac{1}{h}\int_{\Rnum^d}\mu_h(x)(\mathcal{\bar A}^h-\mathcal{A})f(x)dx
= -\lim_{h\rightarrow 0}\int_{\Rnum^d}\mu_h(x)\mathcal{L}^hf(x)dx.
\end{split}
\end{equation*}
It is easy to see that
\begin{equation*}
\int_{\Rnum^d}\mu_h(x)\mathcal{L}^hf(x)dx = \int_{\Rnum^d}(\mu_h(x)-\mu(x))\mathcal{L}^hf(x)dx +\int_{\Rnum^d}\mu(x)\mathcal{L}^hf(x)dx := \mbox{I}+\mbox{II}.
\end{equation*}
By Assumptions \ref{ass} and \ref{differentiable}, we have
\begin{equation*}
\mbox{I}\leq\|\mu_h-\mu\|_{L^1(\Rnum^d)}\left\|\mathcal{L}^hf\right\|\rightarrow 0,\;\;\;
\mbox{as}\;h\rightarrow 0.
\end{equation*}
On the other hand, it follows from the dominated convergence theorem that
\begin{equation*}
\mbox{II}\rightarrow \int_{\Rnum^d}\mu(x)\mathcal{L}f(x)dx,\;\;\;\mbox{as}\;h\rightarrow 0,
\end{equation*}
which gives the desired result.
\end{proof}

We are now in a position to prove the Seifert-Speck-type FDT. Recall that we always assume the regular conditions (a)-(e) and Assumption \ref{ass} to be satisfied.

\begin{theorem}\label{SS}
Let $w$ be the function on $\Rnum^d$ defined by
\begin{equation*}
w(x) = \frac{\nu(x)}{\mu(x)}.
\end{equation*}
Then for any $f\in D(\mathcal{A})\cap C^{2+\theta}_b(\Rnum^d)$ and $0\leq s\leq t$,
\begin{equation}\label{FDTSS}
R_f(s,t) = \frac{\partial}{\partial s}\Enum f(X_t)w(X_s).
\end{equation}
\end{theorem}

\begin{proof}
By the definition of $w(x)$, we have
\begin{equation*}
\begin{split}
\Enum f(X_t)w(X_s) &= \Enum w(X_s)\Enum_{X_s}f(X_{t-s}) = \int_{\Rnum^d}\Enum_xf(X_{t-s})w(x)\mu(x)dx
= \int_{\Rnum^d}P_{t-s}f(x)\nu(x)dx.
\end{split}
\end{equation*}
Recall the following property of the weak generator \cite[Lemma 2.3.3]{lorenzi2006analytical}: for any $f\in D(\mathcal{A})$, we have $P_tf\in D(\mathcal{A})$ and
\begin{equation}\label{basic}
\frac{d}{dt}P_tf(x) = \mathcal{A}P_tf(x) = P_t\mathcal{A}f(x),\;\;\;\forall\;t\geq 0,x\in\Rnum^d.
\end{equation}
By the dominated convergence theorem, we have
\begin{equation*}
\frac{\partial}{\partial s}\Enum f(X_t)w(X_s)
= \int_{\Rnum^d}\frac{\partial}{\partial s}P_{t-s}f(x)\nu(x)dx
= -\int_{\Rnum^d}\mathcal{A}P_{t-s}f(x)\nu(x)dx.
\end{equation*}
It follows from Lemma \ref{contractive} that $P_{t-s}f\in D(\mathcal{A})\cap C^2_b(\Rnum^d)$. This fact, together with Theorem \ref{expression} and Lemma \ref{relation}, shows that
\begin{equation*}
\frac{\partial}{\partial s}\Enum f(X_t)w(X_s)
= \int_{\Rnum^d}\mathcal{L}P_{t-s}f(x)\mu(x)dx = R_f(s,t),
\end{equation*}
which gives the desired result.
\end{proof}

The function $w$ in the above theorem is called the \emph{conjugate observable} in the physics literature. The following theorem shows that under mild conditions, the conjugate observable in the Seifert-Speck-type FDT is unique up to a constant.

\begin{theorem}\label{uniqueness2}
Assume that there exists $K>0$ such that
\begin{equation*}
|b(x)|\leq K(1+|x|),\;\;\;|a(x)|\leq K(1+|x|^2),\;\;\;\forall\; x\in\Rnum^d.
\end{equation*}
Assume that there exists another function $\tilde w\in L^1(\mu)$ on $\Rnum^d$ such that
\begin{equation}\label{equniquness}
\frac{\partial}{\partial s}\Enum f(X_t)w(X_s)
= \frac{\partial}{\partial s}\Enum f(X_t)\tilde{w}(X_s),\;\;\;
\forall\;f\in C_c^{\infty}(\Rnum^d),0\leq s<t\leq T.
\end{equation}
Then $w-\tilde w$ must be a constant almost everywhere.
\end{theorem}

\begin{proof}
From \eqref{equniquness}, it is easy to check that
\begin{equation*}
\frac{\partial}{\partial s}\int_{\Rnum^d} P_{t-s}f(x)(w(x)-\tilde w(x))\mu(x)dx = 0.
\end{equation*}
Since $f\in C_c^{\infty}(\Rnum^d)\subset D(\mathcal{A})$, it follows from the dominated convergence theorem that
\begin{equation*}
\int_{\Rnum^d}\mathcal{A}f(x)(w(x)-\tilde w(x))\mu(x)dx = 0.
\end{equation*}
Since $w-w\in L^1(\mu)$, it follows from \cite[Proposition 4.3.6 and Theorem 4.3.3]{bogachev2015fokker}
that $w-\tilde{w}$ is a constant almost everywhere.
\end{proof}

\begin{remark}
In the Seifert-Speck-type FDT, if we allow the conjugate observable $w_s$ to depend on the early time $s$, then it can be proved that there will be an infinite number of conjugate observables satisfying \eqref{FDTSS} and thus \emph{the uniqueness will be broken}. However, according to Theorem \ref{uniqueness1}, the Agarwal-type FDT has a unique conjugate observable $v_s$ even if we allow it to depend on $s$. This is an important difference between the two types of FDTs.
\end{remark}

To understand the physical implication of the Seifert-Speck-type FDT, let us recall the following concept from stochastic thermodynamics \cite{seifert2010fluctuation}.

\begin{definition}
The \emph{stochastic entropy} of the stationary density $\mu_h$ is an observable $s_h: \Rnum^d\rightarrow\Rnum$ defined as
\begin{equation*}
s_h(x) = -\log\mu_h(x).
\end{equation*}
By Proposition \ref{proregularityinvariant}(a), it is clear that $\mu_h$ is always positive and thus $s_h$ is well defined.
\end{definition}

\begin{remark}
If $\mu_h$ is differentiable with respect to $h$ in the usual sense, then $\nu = \partial_h|_{h=0}\mu_h$ almost everywhere and
\begin{equation*}
w(x) = \frac{\partial_h|_{h=0}\mu_h(x)}{\mu(x)} = -\partial_h|_{h=0}s_h(x).
\end{equation*}
Moreover, it follows from Assumption \ref{differentiable}(b) that
\begin{equation*}
\Enum w(X_s) = \int_{\Rnum^d}\nu(x)dx =
\lim_{h\rightarrow 0}\int_{\Rnum^d}\frac{1}{h}(\mu_h(x)-\mu(x))dx = 0.
\end{equation*}
Therefore, the Seifert-Speck-type FDT shows that for homogeneous and stationary diffusion processes, the
response of an observable to a small external perturbation can be expressed as the correlation function of this observable and another one that is conjugate to the perturbation with respect to stochastic entropy.
\end{remark}

\section{Relationship between the two types of FDTs}
When $X$ is homogeneous and stationary, we have proved two types of FDTs as stated in Theorems \ref{Agarwal2} and \ref{SS}:
\begin{equation}
R_f(s,t) = \Enum f(X_t)v(X_s) = \frac{\partial}{\partial s}\Enum f(X_t)w(X_s).
\end{equation}
Readers may ask what is the connection between the conjugate observables $v$ and $w$. Assume that the conditions of the two types of FDTs are both satisfied. Then Lemma \ref{relation} shows that for any $f\in D(\mathcal{A})\cap C^2_b(\Rnum^d)$,
\begin{equation*}
\int_{\Rnum^d}\nu(x)\mathcal{A}f(x)dx = -\int_{\Rnum^d}\mu(x)\mathcal{L}f(x)dx = -\int_{\Rnum^d}\mathcal{L}^*\mu(x)f(x)dx.
\end{equation*}
By the definitions of $v$ and $w$, we have
\begin{equation*}
(\mathcal{A}f,w)_{\mu} = \int_{\Rnum^d}\mu(x)w(x)\mathcal{A}f(x)dx = -\int_{\Rnum^d}\mu(x)v(x)f(x)dx = -(f,v)_{\mu},
\end{equation*}
where $(\cdot,\cdot)_{\mu}$ is the inner product of two functions with respect to the stationary density $\mu$. This shows that the conjugate observables $v$ and $w$ are formally related by
\begin{equation}\label{connection}
v = -\mathcal{A}^\dag w,
\end{equation}
where $\mathcal{A}^\dag$ is the adjoint operator of $\mathcal{A}$ with respect to the inner product $(\cdot,\cdot)_{\mu}$.

The operator $\mathcal{A}^\dag$ can be understood in two different ways. Under mild conditions, it can be proved that the time-reversed process of $X$ is also a homogenous and stationary diffusion process with drift $b^\dag = -b+\nabla a+a\nabla\log\mu$ and diffusion matrix $a^\dag = a$ \cite[Theorem 3.3.5]{jiang2004mathematical}. In fact, the generator of the time-reversed process is exactly the operator $\mathcal{A}^\dag$, which can be written as
\begin{equation*}
\mathcal{A}^\dag = \sum_{i=1}^d(-b^i+\partial_ja^{ij}+a^{ij}\partial_j\log\mu)\partial_i
+\frac{1}{2}\sum_{i,j=1}^da^{ij}\partial_{ij}.
\end{equation*}
From the perspective of Nelson's stochastic mechanics \cite{nelson1967dynamical, nelson1985quantum}, the mean backward velocity of an observable $f$ is another observable $V_{\textrm{backward}}f$ defined as
\begin{equation*}
V_{\textrm{backward}}f(x) = \lim_{h\downarrow 0}\frac{1}{h}\Enum\{f(X_t)-f(X_{t-h})|X_t=x\}.
\end{equation*}
Under mild conditions, the mean backward velocity of $f$ can be written as \cite[Section 4.2.1]{jiang2004mathematical}
\begin{equation*}
V_{\textrm{backward}}f = -\mathcal{A}^\dag f.
\end{equation*}
Therefore, it follows from \eqref{connection} that the conjugate observable in the Agarwal-type FDT is exactly the mean backward velocity of that in the Seifert-Speck-type FDT. This builds up a bridge between the two types of FDTs.

\section{Examples}
The classical theory of parabolic equations can only deal with the case of bounded drift and diffusion matrix. Here we show how our theory can be applied to diffusion processes with unbounded drift or diffusion coefficients.

\subsection{Agarwal-type FDT for inhomogeneous OU processes}
The classical OU process describes the velocity of an underdamped Brownian particle or the position of an overdamped Brownian particle driven by the harmonic potential \cite{uhlenbeck1930theory}. Here we consider the following $d$-dimensional \emph{inhomogeneous OU process} $X = \{X_t:t\geq 0\}$, which is the solution to the following SDE:
\begin{equation}\label{OU}
dX_t = (B(t)X_t+g(t))dt+A(t)dW_t,\;\;\;X_0 = x_0,
\end{equation}
whose drift $b = b(t,x)$ and diffusion matrix $a = a(t)$ are given by
\begin{equation*}
b(t,x) = B(t)x+g(t),\;\;\;a(t) = A(t)A(t)^T,
\end{equation*}
where $g: \Rnum^+\rightarrow\Rnum^d$, $B:\Rnum^+\rightarrow M_{d\times d}(\Rnum)$, and $A:\Rnum^+\rightarrow M_{d\times n}(\Rnum)$ are continuous. We further assume that $a$ satisfies the following strictly elliptic condition: there exists $\lambda>0$ such that
\begin{equation}\label{strong}
\xi^Ta(t)\xi\geq \lambda|\xi|^2,\;\;\;\forall\;t\geq 0,\xi\in\Rnum^d.
\end{equation}
Inhomogeneous OU processes are also important models in statistical physics \cite{van2003extension}.

\begin{lemma}
$X$ satisfies the regular conditions.
\end{lemma}

\begin{proof}
The regular condition (b) follows from the strictly elliptic condition with $\eta(t,x) = \lambda$. Since $b$ is linear with respect to $x$ and $a$ is independent of $x$, it is easy to check that the regular conditions (a),(c), and (d) hold. If we take $\psi(x) = 1+|x|^2$, then for any $0\leq t\leq T$,
\begin{equation*}
\begin{split}
\mathcal{A}_t\psi(x) &= 2b(t,x)^Tx+\tr(a(t)) = 2x^TB(t)x+2g(t)^Tx+\tr(a(t))\\
&\leq 2\|B\|_{C[0,T]}|x|^2+2\|g\|_{C[0,T]}|x|+\|a\|_{C[0,T]}.
\end{split}
\end{equation*}
Since $|x|\leq \max\{1,|x|^2\}$, the regular condition (e) also holds.
\end{proof}

For the inhomogeneous OU process, it is convenient to introduce the following notations. For any $s,t\in\Rnum$, let $T(s,t)\in M_{d\times d}(\Rnum)$ denote the solution to the following matrix-valued ordinary differential equation (ODE):
\begin{equation}\label{flow}
\dot{x} = B(t)x,\;\;\;x(s) = I.
\end{equation}
Then the solution of \eqref{OU} can be calculated explicitly as \cite{geissert2008asymptotic}
\begin{equation*}
X_t = T(0,t)x_0+\int_0^tT(s,t)g(s)ds+\int_0^tT(s,t)A(s)dW_s.
\end{equation*}
This indicates that $X_t$ is a Gaussian random variable for any $t>0$ with mean
\begin{equation*}
m(t) = T(0,t)x_0+\int_0^tT(s,t)g(s)ds
\end{equation*}
and covariance matrix
\begin{equation*}
\Sigma(t) = \int_0^tT(s,t)a(s)T(s,t)^Tds.
\end{equation*}
Therefore, $X_t$ has a probability density $p_t\in C^\infty_b(\Rnum^d)$ which is given by
\begin{equation}\label{normal}
p_t(x) = [2\pi\det(\Sigma(t))]^{-\frac{1}{2}}e^{-\frac{1}{2}(x-m(t))^T\Sigma(t)^{-1}(x-m(t))}.
\end{equation}

We next consider the perturbed process $X^h = \{X^h_t:t\geq 0\}$ whose drift $b_h = b_h(t,x)$ and diffusion matrix $a_h = a_h(t)$ are given by
\begin{equation*}
b_h(t,x) = b(t,x)+hq_h(t,x),\;\;\;a_h(t,x) = a(t)+hr_h(t,x),
\end{equation*}
where $q_h$ and $r_h$ satisfy Assumption \ref{ass}. We assume that $b_h$ and $a_h$ are differentiable with respect to $h$ and write
\begin{equation*}
q(t,x) = \partial_h|_{h=0}b_h(t,x),\;\;\;r(t,x) = \partial_h|_{h=0}a_h(t,x).
\end{equation*}

For convenience, set $m(s) = (m^i(s))$ and $\Sigma^{-1}(s) = (\sigma_{ij}(s))$ for any $s\geq 0$. The following theorem gives the Agarwal-type FDT for inhomogeneous OU processes.

\begin{theorem}\label{AgarwalOU}
Fix $0 < s\leq t\leq T$. Assume $q(s,\cdot)\in C_b^1(\Rnum^d)$ and $r(s,\cdot)\in C_b^2(\Rnum^d)$. Let $v_s$ be a function on $\Rnum^d$ defined by
\begin{equation*}
\begin{split}
v_s(x) = &\; -\sum_{i=1}^d\partial_iq^i(s,x)
+\frac{1}{2}\sum_{i,j=1}^d\partial_{ij}r^{ij}(s,x)-r^{ij}(s,x)\sigma_{ij}(s)\\
&\; +\sum_{i,j,k=1}^d(q^i(s,x)-\partial_jr^{ij}(s,x))\sigma_{ik}(s)(x^k-m^k(s))\\
&\; +\frac{1}{2}\sum_{i,j,k,l=1}^dr^{ij}(s,x)\sigma_{ik}(s)\sigma_{jl}(s)(x^k-m^k(s))(x^l-m^l(s)).
\end{split}
\end{equation*}
Then for any $f\in C^{2+\theta}_b(\Rnum^d)$,
\begin{equation*}
R_f(s,t) = \Enum f(X_t)v_s(X_s).
\end{equation*}
\end{theorem}

\begin{proof}
Since the probability density $p_s$ exponentially decays with respect to $x$, it is easy to check that the assumptions of Theorem \ref{Agarwal1} can be weakened as $q(s,\cdot)\in C_b^1(\Rnum^d)$ and $r(s,\cdot)\in C_b^2(\Rnum^d)$. By Theorem \ref{Agarwal1}, we have $R_f(s,t) = \Enum f(X_t)v_s(X_s)$, where $v_s = \mathcal{L}^*_sp_s/p_s$. It is easy to see that
\begin{equation*}
\mathcal{L}^*_sp_s = -\partial_i(q^ip_s)+\frac{1}{2}\partial_{ij}(r^{ij}p_s)
= -(\partial_iq^i-\frac{1}{2}\partial_{ij}r^{ij})p_s
-(q^i-\partial_jr^{ij})\partial_ip_s+\frac{1}{2}r^{ij}\partial_{ij}p_s,
\end{equation*}
where we have used Einstein's summation convention: if the same index appears twice in any term, once as an upper index and once as a lower index, that term is understood to be summed over all possible values of that index. From \eqref{normal}, it is easy to check that
\begin{equation*}
\begin{split}
& \partial_ip_s = -\sigma_{ik}(x^k-m^k)p_s,\\
& \partial_{ij}p_s = [-\sigma_{ij}+\sigma_{ik}\sigma_{jl}(x^k-m^k)(x^l-m^l)]p_s.
\end{split}
\end{equation*}
Thus we finally obtain that
\begin{equation*}
\begin{split}
v_s = &\; -(\partial_iq^i-\frac{1}{2}\partial_{ij}r^{ij})+(q^i-\partial_jr^{ij})\sigma_{ik}(x^k-m^k)\\
&\; +\frac{1}{2}r^{ij}[-\sigma_{ij}+\sigma_{ik}\sigma_{jl}(s)(x^k-m^k)(x^l-m^l)]\\
= &\; -\partial_iq^i+\frac{1}{2}\partial_{ij}r^{ij}-\frac{1}{2}r^{ij}\sigma_{ij}
+(q^i-\partial_jr^{ij})\sigma_{ik}(x^k-m^k)\\
&\; +\frac{1}{2}r^{ij}\sigma_{ik}\sigma_{jl}(x^k-m^k)(x^l-m^l),
\end{split}
\end{equation*}
which gives the desired result.
\end{proof}

\subsection{Agarwal-type FDT for homogeneous OU processes}
As a special case, we consider the following $d$-dimensional \emph{homogeneous OU process} $X = \{X_t:t\geq 0\}$, which is the solution to the following SDE:
\begin{equation*}
dX_t = (BX_t+g)dt+AdW_t,
\end{equation*}
whose drift $b = b(x)$ and diffusion matrix $a$ are given by
\begin{equation*}
b(x) = Bx+g,\;\;\;a = AA^T,
\end{equation*}
where $g\in\Rnum^d$, $B\in M_{d\times d}(\Rnum)$, and $A\in M_{d\times n}(\Rnum)$. The following lemma gives the sufficient and necessary condition for the existence of a stationary distribution.

\begin{lemma}\label{eigenvalues}
The stationary distribution of $X$ exists if and only if all the eigenvalues of $B$ have negative real parts. The stationary distribution of $X$, if it exists, must be a Gaussian distribution with mean $m = -B^{-1}g$ and covariance matrix
\begin{equation*}
\Sigma = \int_0^\infty e^{sB}ae^{sB^T}ds.
\end{equation*}
In other words, the stationary density $\mu$ of $X$ is given by
\begin{equation*}
\mu(x) = [2\pi\det(\Sigma)]^{-\frac{1}{2}}e^{-\frac{1}{2}(x-m)^T\Sigma^{-1}(x-m)}.
\end{equation*}
\end{lemma}

\begin{proof}
It is a classical result that the lemma holds when $g = 0$ \cite[Proposition 9.3.1 and Remark 9.3.2]{lorenzi2006analytical}. The proof in general case of $g\neq 0$ is straightforward by using the method of translation.
\end{proof}

For convenience, set $m = (m^i)$ and $\Sigma^{-1} = (\sigma_{ij})$. The following theorem, which is a direct corollary of Theorem \ref{AgarwalOU}, gives the Agarwal-type FDT for homogeneous OU processes.

\begin{theorem}
Assume $q\in C_b^1(\Rnum^d)$ and $r\in C_b^2(\Rnum^d)$. Let $v$ be a function on $\Rnum^d$ defined by
\begin{equation*}
\begin{split}
v(x) = &\; -\sum_{i=1}^d\partial_iq^i(x)
+\frac{1}{2}\sum_{i,j=1}^d\partial_{ij}r^{ij}(x)-r^{ij}(x)\sigma_{ij}
+\sum_{i,j,k=1}^d(q^i(x)-\partial_jr^{ij}(x))\sigma_{ik}(x^k-m^k)\\
&\; +\frac{1}{2}\sum_{i,j,k,l=1}^dr^{ij}(x)\sigma_{ik}\sigma_{jl}(x^k-m^k)(x^l-m^l).
\end{split}
\end{equation*}
Then for any $0\leq s\leq t$ and $f\in C^{2+\theta}_b(\Rnum^d)$,
\begin{equation*}
R_f(s,t) = \Enum f(X_t)v(X_s).
\end{equation*}
\end{theorem}

\subsection{Seifert-Speck-type FDT for homogeneous OU processes}
In this section, we still focus on the homogenous OU process $X$. For simplicity of calculation, we assume that $B$ is a symmetric matrix whose all eigenvalues are negative. Moreover, we assume that the perturbed drift $b_h$ and diffusion matrix $a_h$ have the form of
\begin{equation*}
b_h(x) = Bx+g+h\tilde g,\;\;\;a_h(x) = a+h\tilde a,
\end{equation*}
where $\tilde g\in\Rnum^d$ and $\tilde a\in M_{d\times d}(\Rnum)$. If $h = h(t)$ is taken as a continuous function, then $b_h$ and $a_h$ correspond to the perturbed process $X^h$. If $h$ is taken as a constant, then $b_h$ and $a_h$ correspond to the auxiliary process $\bar X^h$.

It is easy to see that the auxiliary process $\bar X^h$ is also a homogenous OU process whose stationary distribution is a Gaussian distribution with mean $m_h = -B^{-1}(g+h\tilde g)$ and covariance matrix
\begin{equation*}
\Sigma_h = \int_0^\infty e^{sB}(a+h\tilde a)e^{sB^T}ds.
\end{equation*}
Thus the stationary density $\mu_h$ of $\bar X^h$ is given by
\begin{equation*}
\mu_h(x) = [2\pi\det(\Sigma_h)]^{-\frac{1}{2}}e^{-\frac{1}{2}(x-m_h)^T\Sigma_h^{-1}(x-m_h)}.
\end{equation*}

\begin{lemma}
When $h$ is sufficiently small, the auxiliary processes $\bar X_h$ satisfies Assumptions \ref{ass} and \ref{differentiable}.
\end{lemma}

\begin{proof}
Since $q_h = \tilde g$ and $r_h = \tilde a$ do not depend on $x$, it is easy to check that Assumptions \ref{ass} and \ref{differentiable}(a) hold. We next use Lemma \ref{criterion} to verify Assumption \ref{differentiable}(b). When $h$ is sufficiently small, we have
\begin{equation*}
a_h^{-1} = (I+ha^{-1}\tilde a)^{-1}a^{-1} = \sum_{n=0}^\infty(-ha^{-1}\tilde a)^na^{-1},
\end{equation*}
which implies that
\begin{equation*}
|a_h^{-1}| \leq \frac{|a^{-1}|}{1-|h||a^{-1}\tilde a|} \leq \frac{1}{2}|a^{-1}|.
\end{equation*}
This shows that $a_h, a_h^{-1}$, and $\partial_ha_h$ are locally bounded. Moreover, it is easy to see that
\begin{equation*}
|b_h(x)|+|\partial_hb_h(x)| \leq |B||x|+|g|+2|\tilde g|,\;\;\;\forall\;h\in[-1,1],x\in\Rnum^d.
\end{equation*}
Since $B$ is a symmetric matrix whose all eigenvalues are negative, we have
\begin{equation*}
b_h(x)^Tx = x^TBx+g^Tx+h\tilde g^Tx \leq -\gamma|x|^2+(|g|+|\tilde g|)|x|,\;\;\;\forall\;h\in[-1,1],x\in\Rnum^d,
\end{equation*}
where $-\gamma$ is the maximum eigenvalue of $B$. Thus all the conditions of Lemma \ref{criterion} are satisfied, which shows that Assumption \ref{differentiable}(b) holds.
\end{proof}

The following theorem gives the Seifert-Speck-type FDT for homogeneous OU processes.

\begin{theorem}
Let $w$ be a function on $\Rnum^d$ defined by
\begin{equation*}
\begin{split}
w(x) =&\; -\pi(2\pi\det(\Sigma))^{-1}\partial_h|_{h=0}\det(\Sigma_h)
-\tilde g^TB^{-T}\Sigma^{-1}(x-m)\\
&\; -\frac{1}{2}(x-m)^T\partial_h|_{h=0}\Sigma_h^{-1}(x-m).
\end{split}
\end{equation*}
Then for any $0\leq s\leq t$ and $f\in D(\mathcal{A})\cap C^{2+\theta}_b(\Rnum^d)$,
\begin{equation*}
R_f(s,t) = \frac{\partial}{\partial s}\Enum f(X_t)w(X_s).
\end{equation*}
\end{theorem}

\begin{proof}
Since Assumptions \ref{ass} and \ref{differentiable} are satisfied, it is easy to check that
\begin{equation*}
\begin{split}
\nu(x) = \partial_h|_{h=0}\mu_h
=&\; \Big[-\pi(2\pi\det(\Sigma))^{-1}\partial_h|_{h=0}\det(\Sigma_h)
+\partial_h|_{h=0}m_h^T\Sigma^{-1}(x-m)\\
&\; -\frac{1}{2}(x-m)^T\partial_h|_{h=0}\Sigma_h^{-1}(x-m)\Big]\mu(x).
\end{split}
\end{equation*}
Since $\partial_h|_{h=0}m_h = -B^{-1}\tilde g$, the desired result follows from Theorem \ref{SS}.
\end{proof}

\section{Conclusions and discussion}
In the present paper, we provide the rigorous mathematical foundation of two types of nonequilibrium FDTs for inhomogeneous diffusion processes with unbounded coefficients. In a previous work, Dembo and Deuschel \cite{dembo2010markovian} have also developed the mathematical theory of nonequilibrium FDTs for homogenous Markov processes within an abstract framework. Since both the two papers investigate FDTs from the mathematical perspective, we feel it necessary to discuss the similarities and differences between our work and the work of Dembo and Deuschel in detail. Although \cite{dembo2010markovian} only states the Seifert-Speck-type FDT explicitly, the Agarwal-type FDT could also be derived from their expression of the response function without much difficulty.

On the other hand, there are three major differences between the two papers. First, \cite{dembo2010markovian} focused on FDTs for homogenous processes and homogenous external perturbations based on the methods of strongly continuous semigroups and Dirichlet forms, whereas our work focuses on FDTs for inhomogeneous processes and inhomogeneous external perturbations based on the  tools of weakly continuous semigroups and Schauder estimates for parabolic equations. Since Dembo and Deuschel only considered the homogenous case, their theory could be widely applied to general Markov processes on complete separable metric spaces. However, since our work concentrates on the inhomogeneous case, which is much more complicated, our model is only restricted to diffusion processes on Euclidean spaces.
Second, \cite{dembo2010markovian} mainly focused on two generic families of perturbations, the so-called time change and the generalized Langevin dynamics, while our work applies to general nonlinear perturbations whenever the perturbed processes still belong to the family of diffusion processes. Third, since the theory in \cite{dembo2010markovian} is established in an abstract setting, the assumptions required for the FDTs are also rather abstract. However, since our work only focuses on diffusion processes, all the assumptions required for the FDTs are imposed on the drift and diffusion matrix and thus are very easy to verify.

There are several technical difficulties in the present work, which are summarized as follows. In the physics literature \cite{marconi2008fluctuation}, the FDTs are usually studied based on the evolution of the probability density, which is governed by the Fokker-Planck equation. However, these methods are quite formal and often involve some vague concepts such that the exponent of the Fokker-Planck operator. To prove the FDTs with full mathematical rigor, we focus on the inhomogeneous transition semigroup $\{P_{s,t}\}$, whose evolution is governed by the Kolmogorov backward equation, a partial differential equation of parabolic type.

For the Agarwal-type FDT, a crucial step is to prove the uniform boundedness and pointwise convergence for the first and second partial derivatives of the perturbed semigroup $\{P^{\epsilon\phi}_{s,t}\}$. Here we overcome this difficulty using the Schauder estimates for parabolic equations \cite{lunardi1998schauder, lorenzi2011optimal}. Recent advances in Schauder estimates allows us to establish the Agarwal-type FDT for inhomogeneous diffusion processes with unbounded coefficients. Another minor challenge encountered is to prove the finite-order weak differentiability of the time-dependent and stationary probability densities of inhomogeneous diffusion processes under the regular conditions. In fact, the first-order weak differentiability of the solution to the Kolmogrov backward equation has been widely studied \cite{bogachev2015fokker}. In this work, we elevate the smoothness to higher orders using the weak solution theory of elliptic and parabolic equations. During the proof, we also find that the original observable in the Agarwal-type FDT should belong to the class of $C^{2+\theta}_b(\Rnum^n)$, rather than $C^2_b(\Rnum^n)$, due to the requirements of Schauder estimates. This indicates that second differentiability of the observable may not suffice to guarantee the Agarwal-type FDT.

Furthermore, we clarify that the Seifert-Speck-type FDT only holds for homogenous and stationary diffusion processes. When the drift and diffusion matrix are unbounded, the transition semigroup $\{P_t\}$ of a homogenous diffusion process may not be strongly continuous on $C_b(\Rnum^d)$ \cite{daprato1995ornstein} and thus the classical semigroup theory is not applicable. Here we overcome this difficulty using the theory of weakly continuous semigroups and give a rigorous proof of the Seifert-Speck-type FDT for homogenous diffusion processes with unbounded coefficients.

In addition, we derive an explicit formula of the response function that applies to any forms of inhomogeneous and nonlinear external perturbations, rather than merely homogenous \cite{dembo2010markovian} or linear \cite{marconi2008fluctuation} perturbations as in most previous papers. The uniqueness of the conjugate observables of the two types of FDTs is also clarified: the conjugate observable in the Agarwal-type FDT is always unique, while the conjugate observable in the Seifert-Speck-type FDT is unique only when it is independent of the early time $s$. When the process is homogeneous and stationary, we also build up a bridge between the two types of FDTs using concepts in Nelson's stochastic mechanics \cite{nelson1967dynamical, nelson1985quantum}. We make it clear that the conjugate observable in the Agarwal-type FDT is exactly the mean backward velocity of that in the Seifert-Speck-type FDT.

Finally, we hope that the nonequilibrium FDTs established in this paper could be extensively applied to study various dissipative dynamic phenomena in physics and biology such as overshoot \cite{jia2014overshoot}, adaptation \cite{lan2012energy, jia2016nonequilibrium}, and the response \cite{mangan2003structure} and relaxation \cite{jia2018relaxation} kinetics of gene regulatory networks. We also anticipate that recent advances in stochastic processes and partial differential equations could facilitate the development of the mathematical foundation for nonequilibrium stochastic thermodynamics \cite{jarzynski2011equalities, seifert2012stochastic, van2015ensemble}.

\section*{Acknowledgments}
The authors are grateful to Professor M. R\"{o}ckner for pointing out the reference \cite{bogachev2015fokker} and also grateful to the anonymous referees for their valuable comments and suggestions which helped us greatly in improving the quality of this paper. X. Chen was supported by National Natural Science Foundation of China (Grant No. 11701483).

\setlength{\bibsep}{5pt}
\small\bibliographystyle{nature}
%\bibliography{FDT}

\end{document}